\documentclass[10pt]{amsart}
\usepackage[latin1]{inputenc}
\usepackage{amsart-color}
\usepackage{hyperref}

\numberwithin{equation}{section}

\newtheorem{theorem}{Theorem}[section]

\newtheorem{proposition}[theorem]{Proposition}

\theoremstyle{definition}
\newtheorem{definition}[theorem]{Definition}
\newtheorem*{remark}{Remark}

\def\emph#1{#1}

\def\pd#1#2{\frac{\partial#1}{\partial#2}}

\newcommand{\at}[1]{\Big\vert_{#1}}

\newcommand{\pai}[2]{\langle\,#1\,,#2\,\rangle}  
\newcommand{\set}[2]{\left\{\,#1\left.\vphantom{#1#2}\,\right\vert\,#2\,\right\}} 
\newcommand{\Real}{\mathbb{R}}
\newcommand{\cinfty}[1]{C^\infty(#1)}
\newcommand{\map}[3]{#1\colon#2\rightarrow#3}    
\newcommand{\vectorfields}[1]{\mathfrak{X}(#1)}
\newcommand{\Sec}[2][]{\operatorname{Sec}\nolimits_{#1}(#2)}
\let\sec\Sec

\DeclareMathOperator{\pr}{pr}
\DeclareMathOperator{\id}{id}

\newcommand{\T}{\mathcal{T}}
\newcommand{\X}{\mathcal{X}}
\newcommand{\V}{\mathcal{V}}

\newcommand{\PJ}[1][E]{\mathcal{P}(J,#1)}
\newcommand{\Ti}[1][]{\boldsymbol{T}^{#1}}
\newcommand{\FL}{\mathcal{F}_L}
\newcommand{\g}{\mathfrak{g}} 
\newcommand{\G}{\boldsymbol{G}} 
\renewcommand{\t}{\boldsymbol{t}} 
\newcommand{\s}{\boldsymbol{s}} 
\def\E{\mathbb{E}} 
\def\S{\mathbb{S}} 
\def\p{\mathsf{p}} 
\def\ad{\operatorname{ad}} 
\def\d{\mathrm{d}}

\def\cala{\mathcal{A}}
\def\calf{\mathcal{F}}
\def\cals{\mathcal{S}}

\let\Info\thanks
\def\Info#1{\relax}

\let\stdparagraph\paragraph
\def\paragraph#1{\par\medskip\stdparagraph{#1}}

\begin{document}


\Info{Submitted to AIMS Journal of Geometric Mechanics on August 21st 2014}

\title{Higher-order Variational Calculus on Lie algebroids}

\author[E.\ Mart\'{\i}nez]{Eduardo Mart\'{\i}nez}
\address{Eduardo Mart\'{\i}nez:
IUMA-Departamento de Matem\'atica Aplicada,
Universidad de Zaragoza,
Pedro Cerbuna 12,
50009 Zaragoza, Spain}
\email{emf@posta.unizar.es}
\thanks{Partial financial support from MINECO (Spain) grant MTM2012-33575, and from Gobierno de Arag\'on (Spain) grant DGA-E24/1 is acknowledged}


\begin{abstract} 
The equations for the critical points of the action functional defined by a Lagrangian depending on higher-order derivatives of admissible curves on a Lie algebroid are found. The relation with Euler-Poincar\'e and Lagrange Poincar\'e type equations is studied. Reduction and reconstruction results for such systems are established.
\end{abstract}

\maketitle


\section{Introduction}

In this paper we study optimization problems defined by a cost functional which depends on higher-order derivatives of admissible curves on a Lie algebroid. Examples of this type of problems are the optimal control of dynamical systems where the system to be controlled is a mechanical system, and hence depends on accelerations~\cite{Ab-quasi,AbCa,CoMaZu,Co}, trajectory planning problems in control theory~\cite{KySa}, key-framed animations in computer graphics~\cite{NoHePa}, and in general, problems of interpolation and approximation of curves on Riemannian manifolds~\cite{CaSiCr, MaSiKr}. In many of these examples the presence of symmetries is used to reduce the difficulty of the problem.

The advantage of the Lie algebroid approach is its inclusive nature, under the same formalism we can describe many systems which are apparently different~\cite{Weinstein, LMLA, NHLSLA,GrGrUr} and hence it allows a unified description even in the case of a reduced system when symmetries are present in the problem. An alternative approach consists in a case by case study using Euler-Poincar{\'e} and Lagrange-Poincar{\'e} reduction techniques as in~\cite{GBHoetal, GBHoRa}.

Previous work on the variational description of first-order Lagrangian systems defined on Lie algebroids provides a convenient departure point for the generalization presented here. In~\cite{VCLA} it was shown that the Euler-Lagrange equations for a Lagrangian system defined on a Lie algebroid are the critical points for the action functional defined on an adequate Banach manifold of admissible curves satisfying boundary conditions. Such a manifold structure is a foliated one, the leaves being  the $E$-homotopy classes of admissible curves.

It is frequently argued that the variational principle for reduced systems, and in general for systems on Lie algebroids, is a constrained variational principle, in the sense that some additional constraints are imposed to the admissible variations. The point of view of~\cite{VCLA} and the present paper is different. The set of admissible curves where the action is defined is endowed with a reasonable Banach manifold structure. The tangent space to such manifold of admissible curves is precisely the whole set of infinitesimal admissible variations. Therefore no additional constraint to the infinitesimal variations is imposed, or in other words, they are as constrained as in the standard case when formulated directly in the tangent bundle instead of on the base manifold. In such standard case the relevant topology on the set $C^1(I,TM)$ is the one induced by the manifold structure of $C^2({\Real},M)$, which is a foliated structure where each leaf (connected component) is a the set of curves of the form $\dot{{\gamma}}$ with ${\gamma}$ in a given homotopy class. A similar construction can be performed in the case of admissible curves on a Lie algebroid $E$ by using the notion of $E$-homotopy~\cite{Rui,VCLA}. 

Whether the reader agrees with the above argument or not, it should be clear that the variational principle stated here is a \textsl{bona fide} standard variational principle, as it is needed in optimization problems: the solutions of the higher-order Euler-Lagrange equations are the critical points of the action functional which is a smooth function on a Banach manifold. In this sense one should notice that some generalized variational principles that appeared in the literature~\cite{GrGr, JoRo} do not satisfy this property.

\subsubsection*{Description of the results and organization of the paper}

For a Lie algebroid $E$ we consider the set $E^k$ of $(k-1)$-jets of admissible curves on $E$. Given a function $L{\in}\cinfty{E^k}$, which will be called the Lagrangian, we want to find the maxima/minima/critical points of a cost/action functional $S$ defined by  
\[
S(a)=\int_{t_0}^{t_1}L(a^k(t))\,dt
\]
among the set of admissible curves which are $E$-homotopic to a given admissible curve $a_0$ and satisfy some boundary conditions. Finite variations $a_s(t)=a(s,t)$ are given by $E$-homotopies and infinitesimal variations are just the $s$-derivative of~$a_s$. $E$-homotopies are special morphisms of Lie algebroids ${\phi}={\alpha}(s,t)dt+{\beta}(s,t)ds$ which satisfy ${\beta}(s,t_0)={\beta}(s,t_1)=0$. This condition corresponds to the fixed endpoint condition in the standard case. The corresponding infinitesimal variation depends only on the values of ${\sigma}(t)={\beta}(0,t)$ and are denoted ${\Xi}^k_{a_0}{\sigma}$. Since each $E$-homotopy class is a Banach manifold~\cite{Rui,VCLA}, we can properly talk about critical points of the function $S$. 

The differential equations satisfied by the critical points are called the higher-order Euler-Lagrange equations, and generalize the first-order Euler-Lagrange equations defined by Weinstein~\cite{Weinstein} (see also~\cite{LMLA,GrGrUr}). It will be shown that, in particular, this equations are the higher-order Euler-Poincar{\'e} equations when the Lie algebroid is a Lie algebra~\cite{GBHoetal}, the higher-order Lagrange-Poincar{\'e} equations when the Lie algebroid is the Atiyah algebroid associated to a principal bundle~\cite{GBHoRa}, or the higher-order Euler-Poincar{\'e} equations with advected parameters when the Lie algebroid is an action algebroid~\cite{GBHoetal,GBHoRa}, in addition to the standard higher-order Euler-Lagrange equations when the Lie algebroid is the tangent bundle to a manifold.

One of the advantages of the formalism of Lie algebroids is that morphisms of Lie algebroids can serve to relate Lie algebroids of the different types mentioned above. When two Lagrangians are related by a morphism of Lie algebroids, the corresponding variational problems are also related and this allows to easily deduce correspondences between the critical points of the associated action functionals, which amounts to reduction theorems showing, among other results, that the standard higher-order Euler-Lagrange equations for a higher-order left-invariant Lagrangian (with/without parameters) on a Lie group reduce to the higher-order Euler-Poincar{\'e} equations of the reduced Lagrangian (with/without advected parameters) on the Lie algebra, or that the standard higher-order Euler-Lagrange equations for a higher-order invariant Lagrangian on a principal bundle reduce to the higher-order Lagrange-Poincar{\'e} equations of the reduced Lagrangian on the Atiyah algebroid.

The paper is organized as follows. In Section~\ref{preliminaries} we will introduce the necessary preliminary results about higher-order tangent bundles and Lie algebroids, and we will fix some notation. In Section~\ref{jets.of.paths} we will define the space of jets of admissible curves and we will study its basic properties. In Section~\ref{variational.vector.fields} we will find the properties and the local expression of the variational vector fields defined by $E$-homotopies, and its relation to the complete lift of a section of $E$. In order to find an intrinsic expression of the Euler-Lagrange equations, we will define in Section~\ref{vertical.endomorphism} two differential operators, the variational operator and the Cartan operator. This will be done by introducing the vertical endomorphism. In Section~\ref{variational.description} we will find the critical points of the action functional in terms of the variational operator and we will deduce its coordinate expression, as well as those of the Cartan form and the Legendre transformation. Also a version of Noether's theorem follows easily from the variational character of the equations. The relevant typical examples are given in~Section~\ref{examples}. Finally, in Section~\ref{reduction.reconstruction} we will study the transformation properties of critical points under morphisms of Lie algebroids which readily amount to reduction results among the different kind of equations for different Lie algebroids.


\section{Preliminaries}
\label{preliminaries}


\subsection{Higher-order tangent bundles}
Let $M$ be a manifold. For a curve $\map{{\gamma}}{{\Real}}{M}$, defined in some open interval containing the origin in ${\Real}$, we denote by $[{\gamma}]^k=j^k_0{\gamma}$ the $k$-jet of ${\gamma}$ at $0$, which is said to be the $k$th-order velocity of ${\gamma}$ or simply the $k$-velocity of ${\gamma}$. The set of $k$-velocities of curves in $M$ is a manifold $T^kM$, known as the $k$th-order tangent manifold to $M$. For $k=1$ we have $T^1M=TM$, the tangent bundle to $M$. The projections ${\tau}^M_{k,l}:[{\gamma}]^k\mapsto[{\gamma}]^l$ define bundles 
\[
T^kM\xrightarrow{{\tau}^M_{k,l}}T^lM\xrightarrow{{\tau}^M_{l,0}} M,\qquad\text{for $k>l>0$}.
\] 
See~\cite{LeRo,CrSaCa} for more information.

A vector tangent to $T^kM$ can be described by a 1-parameter family of curves $\map{{\gamma}}{{\Real}^2}{M}$ defined locally in a neighborhood of the origin in ${\Real}^2$. Concretely, the family ${\gamma}_s(t)={\gamma}(s,t)$ defines a curve in $T^kM$ by fixing $s$ and taking $k$-jets $[{\gamma}_s]^k$. The vector $\frac{d}{ds}[{\gamma}_s]^k\at{s=0}$ tangent to such curve at $s=0$ is tangent to $T^kM$ at the point $[{\gamma}_0]^k$. With a different notation, that will be used in what follows in this paper, it is the vector $[s\mapsto[t\mapsto {\gamma}(s,t)]^k]^1$.

For a curve $\map{{\gamma}}{{\Real}}{M}$ we will denote by ${\gamma}^{(k)}$ the curve $\map{{\gamma}^{(k)}}{{\Real}}{T^kM}$ given by ${\gamma}^{(k)}(t)=[s\mapsto {\gamma}(t+s)]^k$. 

There is a canonical injective immersion $\map{i^M_{k,1}}{T^{k+1}M}{T(T^kM)}$, defined by $i^M_{k,1}([{\gamma}]^{k+1})=[s\mapsto[t\mapsto {\gamma}(t+s)]^k]^1$. Such an immersion allows to identify $(k+1)$-velocities with the vectors tangent to $T^kM$ which are in the image of $i^M_{k,1}$. 

The canonical flip map on $T^kM$ will be denoted by $\map{{\chi}_k}{T^kTM}{TT^kM}$. It can be easily defined in terms of 1-parameter families of curves by 
\begin{equation}
{\chi}_k([t\mapsto[s\mapsto {\gamma}(s,t)]^1]^k)=[s\mapsto[t\mapsto {\gamma}(s,t)]^k]^1.
\end{equation}
Fixing local coordinates $x^i$ in $M$, we have an induced system of local coordinates $(x^i_{(j)})$, $j=0,\ldots,k$, of $T^kM$ given by $x^i_{(j)}([{\gamma}]^k)=\frac{d^j{\gamma}^i}{dt^j}(0)$. 


\subsection{Lie algebroids}
A Lie algebroid structure on a vector bundle $\map{\tau}{E}{M}$ is given by a vector bundle map $\map{\rho}{E}{TM}$ over the identity in $M$, called the anchor, together with a Lie algebra structure on the $\cinfty{M}$-module of sections of $E$ such that the compatibility condition 
$
[\sigma,f\eta]=(\rho(\sigma)f)\eta+f[\sigma,\eta]
$
is satisfied for every $f\in\cinfty{M}$ and every $\sigma,\eta\in\sec{E}$. See~\cite{Mackenzie2} for more information on Lie algebroids.

In what concerns to Variational Calculus and Mechanics, it is convenient to think of a Lie
algebroid as a generalization of the tangent bundle of $M$. One
regards an element $a$ of $E$ as a generalized velocity, and the
actual velocity $v$ is obtained when applying the anchor to $a$, i.e.,
$v=\rho(a)$. A curve $\map{a}{[t_0,t_1]}{E}$ is said to be
admissible, or an $E$-path, if $\dot{\gamma}(t)=\rho(a(t))$, where $\gamma(t)=\tau(a(t))$
is the base curve.

The Lie algebroid structure is equivalent to the existence of a  degree~1 derivation, $\map{d}{\Sec{\wedge^kE^*}}{\Sec{\wedge^{k+1}E^*}}$, which is a cohomology operator $d^2=0$, and is known as the exterior differential on $E$. A morphism of Lie algebroids is a vector bundle map $\map{{\Phi}}{E}{E'}$ such that ${\Phi}^\star{\circ}d=d{\circ}{\Phi}^\star$. We may also define the Lie derivative with respect to a section $\sigma$ of $E$ as the operator $\map{d_\sigma}{\Sec{\wedge^k E^*}}{\Sec{\wedge^k E^*}}$ given by $d_\sigma=i_\sigma\circ d+d\circ i_\sigma$.  Along this paper, the symbol $d$ stands for the exterior differential on a Lie algebroid while the non-slanted symbol $\d$ stands for the standard exterior differential on a manifold. 

A local coordinate system $(x^i)$, $i=1,\ldots,n=\dim(M)$,  in the base manifold $M$ and a local
basis $\{e_{\alpha}\}$, ${\alpha}=1,\ldots,m=\operatorname{rank}(E)$, of sections of $E$ determine a local coordinate system  $(x^i, y^{\alpha})$ on $E$.  The anchor and the bracket are locally determined by the structure functions $\rho^i_\alpha$ and $C^\alpha_{\beta\gamma}$ on $M$ given by
\begin{equation}
\rho (e_{\alpha})=\rho _{\alpha}^{i}\frac{\partial}{\partial x^i}
\qquad\text{and}\qquad
[e_{\alpha}, e_{\beta}]=C_{\alpha\beta}^{\gamma}\ e_{\gamma}.
\end{equation}
The exterior differential $d$ on the Lie algebroid is locally determined by
\begin{equation}
dx^i=\rho^i_{\alpha}e^{\alpha}\qquad\text{and}\qquad de^{\gamma}=-\frac{1}{2}
C^{\gamma}_{\alpha\beta} e^{\alpha}\wedge e^{\beta},
\end{equation}
where $\{e^{\alpha}\}$ is the dual basis of $\{e_{\alpha}\}$. 

\paragraph{The $E$-tangent to a fibration}\cite{LMLA,Medina,LSDLA}. Let $E$ be a Lie algebroid over a manifold $M$ and $\map{{\pi}}{P}{M}$ be a fibration. The \emph{$E$-tangent} to $P$ is the Lie algebroid $\map{{\tau}^E_P}{\T^EP}{P}$ whose fibre at $p$ is the vector space 
\[
\T^E_pP =\set{(b,v)\in E_{{\pi}(p)}\times T_pP}{\rho(b)=T_p\pi(v)},
\]
the anchor is ${\rho}(b,v)=v$ and the bracket is determined by the bracket of projectable sections. We will use the redundant notation $(p,b,v)$ to denote the element $(b,v)\in\T^E_pP$. The projection onto the second factor $(p,b,v)\mapsto b$ is a morphism of Lie
algebroids and will be denoted by $\map{\T{\pi}}{\T^EP}{E}$.

Given local coordinates $(x^i,u^A)$ on $P$ and a local basis
$\{e_\alpha\}$ of sections of $E$, we can define a local basis
$\{\X_\alpha,\V_A\}$ of sections of $\T^EP$ by
\[
\X_\alpha(p)
=\Bigl(p,e_\alpha(\pi(p)),\rho^i_\alpha\pd{}{x^i}\at{p}\Bigr) \qquad\text{and}\qquad
\V_A(p) = \Bigl(p,0_{{\pi}(p)},\pd{}{u^A}\at{p}\Bigr).
\]
Locally, the Lie brackets of the elements of the basis are
\begin{equation}
[\X_\alpha,\X_\beta]= C^\gamma_{\alpha\beta}\:\X_\gamma,
\qquad 
[\X_\alpha,\V_B]=0 
\qquad\text{and}\qquad
[V_A,\V_B]=0,
\end{equation}
and, therefore, the exterior differential is determined by
\begin{equation}
\begin{aligned}
  &dx^i=\rho^i_\alpha \X^\alpha,
  &&du^A=\V^A,\\
  &d\X^\gamma=-\frac{1}{2}C^\gamma_{\alpha\beta}\X^\alpha\wedge\X^\beta,\qquad
  &&d\V^A=0,
\end{aligned}
\end{equation}
where $\{\X^\alpha,\V^A\}$ is the dual basis to $\{\X_\alpha,\V_A\}$.

\paragraph{First-order variational vector fields}
Within the context of Variational Calculus on Lie algebroids, a variation of an admissible curve $\map{a}{[t_0,t_1]}{E}$ is associated to a morphism of Lie algebroids $\map{{\alpha}(s,t)dt+{\beta}(s,t)ds}{T{\Real}^2}{E}$ such that $a(t)={\alpha}(0,t)$ and ${\beta}(s,t_0)=0$, ${\beta}(s,t_1)=0$. The variational vector field $\pd{{\alpha}}{s}(0,t)$ is determined by ${\sigma}(t)={\beta}(0,t)$ and it is denoted ${\Xi}_a{\sigma}(t)$. In local coordinates, it has the expression 
\begin{equation}
\label{first-order.variational.vector.field}
{\Xi}_a({\sigma})(t)={\rho}^i_{\alpha}({\gamma}(t)){\sigma}^{\alpha}(t)\pd{}{x^i}\at{a(t)}
+\Bigl(\dot{{\sigma}}^{\alpha}(t)+C^{\alpha}_{{\beta}{\mu}}({\gamma}(t))a^{\beta}(t){\sigma}^{\mu}(t)\Bigr)\pd{}{y^{\alpha}}\at{a(t)}.
\end{equation}
where $a$ and ${\sigma}$ have local expression $a(t)=(\gamma^i(t),a^\alpha(t))$ and ${\sigma}(t)={\sigma}^{\alpha}(t)e_{\alpha}({\gamma}(t))$. See~\cite{LAGGM,VCLA}.

The variational vector field can also be defined in terms of the canonical involution of the bundle $\T^EE$. See~\cite{LSDLA, VCLA} for the details.


\section{Jets of admissible curves}
\label{jets.of.paths}

Consider a Lie algebroid\footnote{Nearly all the material in this section remains valid for an anchored vector bundle $(\map{{\tau}}{E}{M},{\rho})$} $(\map{{\tau}}{E}{M},{\rho},[\ ,\ ])$. A curve $\map{a}{I{\subset}{\Real}}{E}$ is said to be an \emph{admissible curve} in $E$, or an \emph{$E$-path}, if it satisfies ${\rho}{\circ}a=\dot{{\gamma}}$, where ${\gamma}={\tau}{\circ}a$ is the base curve. 

\begin{definition}
For $k{\in}\mathbb{N}$, we denote by $E^k$ the set of $(k-1)$-jets of admissible curves on $E$:
\[
E^k=\set{[a]^{k-1}{\in}T^{k-1}E}{\text{$a$ is an admissible curve in $E$}}.
\]
\end{definition}

These spaces were introduced by Colombo and Mart{\'\i}n de Diego in~\cite{CoMa-Z}. See also~\cite{JoRo} for a simple exposition.

\begin{remark}
Notice the grading $E^1=E$, $E^2{\subset}TE$, and in general $E^k{\subset}T^{k-1}E$. The notation is suggested by the intended use in higher-order mechanics, where the elements in $E$ are already considered as quasi-velocities (i.e. 1-quasi-velocities are in $E^1\equiv E$). For instance, in the standard case $E=TM$, we have $E^1=TM$, $E^2=T^2M$, etc. In our notation, plain indices will indicate the space where the object is defined, while indices between parenthesis will indicate jet prolongation or the number of derivatives.
\end{remark}

To gain some intuition, it is convenient to describe the situation locally. Taking local coordinates $(x^i,y^{\alpha})$ on $E$, an admissible curve $a(t)=({\gamma}^i(t),a^{\alpha}(t))$ is determined by the function $a^{\alpha}(t)$ and the initial value ${\gamma}^i(0)$, since the function ${\gamma}^i(t)$ can be determined as the solution of the initial value problem $\dot{x}^i={\rho}^i_{\alpha}(x)a^{\alpha}(t)$ with initial condition $x(0)={\gamma}(0)$. Therefore, the $(k-1)$-jet of $a(t)$ corresponds just to the $(k-1)$-jet of the function $a^{\alpha}(t)$ together with the initial value ${\gamma}^i(0)$, the $0$-jet of ${\gamma}^i(t)$. The natural coordinates $\bigl(x^i_{(j)},y^{\alpha}_{(j)}\bigr)$ of $[a]^{k-1}{\in}T^{k-1}E$ are given by 
\begin{equation}
\label{local.coordinates}
\begin{aligned}
x^i_{(0)}&={\gamma}^i(0),&\\
y^{\alpha}_{(r)}&=\frac{d^{r-1}a^{\alpha}}{dt^{r-1}}(0),&r=1,\ldots,k-1,\\
x^i_{(r)}&={\Psi}^i_r\left({\gamma}^i(0),a^{\alpha}(0),\ldots,\frac{d^{r-1}a^{\alpha}}{dt^{r-1}}(0)\right),&r=1,\ldots,k-1,
\end{aligned}
\end{equation}
where ${\Psi}^i_r$ are smooth functions depending also smoothly on ${\rho}^i_{\alpha}$ and its partial derivatives up to order $r-1$, obtained by taking total derivatives of the admissibility condition $\dot{x}^i={\rho}^i_{\alpha}a^{\alpha}$. 

Conversely, given a point $(x^i_0,y^{\alpha}_1,\ldots,y^{\alpha}_k){\in}{\Real}^n\times{\Real}^{k{\cdot}m}$ the admissible curve given by $a^{\alpha}(t)=\sum_{j=0}^{k-1}\frac{1}{j!}y^{\alpha}_{j+1}t^j$ and the solution ${\gamma}^i(t)$ of the initial value problem $\dot{x}^i={\rho}^i_{\alpha}(x)a^{\alpha}(t)$, $x^i(0)=x^i_0$, is an admissible curve whose $(k-1)$-jet has coordinates as given in~\eqref{local.coordinates} with $\frac{d^{r-1}a^{\alpha}}{dt^{r-1}}(0)=y^{\alpha}_r$ for $r=1,\ldots,k-1$. 

It follows that $E^k$ is a smooth submanifold of $T^{k-1}E$ with dimension $\dim E^k=n+km$, and that we can take a local coordinate system $(x^i,y^{\alpha}_r)$ of the form given above, $x^i=x^i_{(0)}$, $y^{\alpha}_r=y^{\alpha}_{(r-1)}$.

We will denote by $\map{{\tau}_{k,l}}{E^k}{E^l}$  the restriction of the natural jet bundle projection $\map{{\tau}^E_{k-1,l-1}}{T^{k-1}E}{T^{l-1}E}$. Then it is straightforward to prove the following result.

\begin{proposition}
The set $E^k$ is a submanifold of $T^kE$. The dimensi{\'o}n of $E^k$ is $\dim(E^k)=\dim(M)+k{\cdot}\operatorname{rank}(E)$. For $k>l>0$ we have that $\map{{\tau}_{k,l}}{E^k}{E^l}$  is a smooth fibre bundle. For $k=l+1$ it is an affine bundle.
\end{proposition}

See Section~\ref{examples} bellow for the construction of $E^k$ for some concrete examples of Lie algebroids.
 
For an admissible curve $\map{a}{{\Real}}{E}$ we will denote by $\map{{a}^k}{{\Real}}{E^k}$ the natural jet-prolongation of $a$ to $E^k$, given by 
\[
{a}^k(t)=[s\mapsto a(s+t)]^{k-1}.
\]
Notice that with the above notations $a^k(t)=a^{(k-1)}(t)$.

\begin{remark}
The manifold $E^{k+1}$ can also be defined as a subset of the $E$-tangent to $E^{k}$ by the following inductive procedure. Starting with $E^1=E$, and once we have constructed $E^k$, we define $E^{k+1}$ as follows.  Consider the submersion $\map{{\tau}_{k,0}}{E^k}{M}$ and the $E$-tangent $\T^EE^k$ to $E^k$ with respect to such projection. Then we define 
\[
E^{k+1}=\set{Z{\in}\T^EE^k}{{\tau}^E_{E^k}(Z)=\T{\tau}_{k,k-1}(Z)}.
\] 
For instance, $E^2=\{(a,a,V)\in\T^EE\}$, is the set of admissible elements, denoted by $\operatorname{Adm}(E)$ in~\cite{LMLA}. This construction allows to consider $E^{k+1}$ as a submanifold $E^{k+1}{\subset}\T^EE^k$. However, it is preferable to consider $E^k$ as a separate manifold, and to define an embedding into $\T^EE^k$ as described in the next paragraph. 
\end{remark}

\paragraph{Canonical inclusion into the $E$-tangent}

There exists a canonical injective immersion map $\map{i_{k,1}}{E^{k+1}}{\T^EE^k}$ defined by 
\[
i_{k,1}([a]^k)=\left([a]^{k-1},[a]^0,[s\mapsto[t\mapsto a(s+t)]^{k-1}]^1\right),
\]
which does not depend on the representative $E$-path $a(t)$ of the $k$-jet $[a]^{k}{\in}E^{k+1}$. 

In terms of the natural jet-prolongation of admissible curves, the canonical immersion $i_{k,1}$ is given by
\[
i_{k,1}({a}^{k+1}(t))=\left({a}^k(t),a(t),\frac{d{a}^k}{dt}(t)\right).
\]

The canonical immersion can be considered as a section of $\T^EE^k$ along ${\tau}_{k+1,k}$. We will use the symbol $\Ti=i_{k,1}$ to denote such a section, that is
\begin{equation}
\label{Ti}
\Ti([a]^k)=\left(a^k(0),a(0),\frac{da^k}{dt}(0)\right).
\end{equation}
For $k=1$ the section $\Ti$ was already introduced in~\cite{LMLA}.

The associated derivation $d_{\Ti}$ maps functions defined on $E^k$ to functions defined on $E^{k+1}$, and will be called the \emph{total time derivative operator}. The reason for that name is that for every function $F{\in}\cinfty{E^k}$ and every admissible curve $a$ we have 
\[
d_{\Ti}F({a}^{k+1}(t))=\frac{d}{dt} F({a}^k(t)).
\]
More generally
\[
d^r_{\Ti}F(a^{k+r}(t))=\frac{d^r}{dt^r}F({a}^k(t)).
\]

Associated to the coordinate system $(x^i,y^{\alpha}_r)$ in $E^k$ we have the local basis $\{\X_{\alpha},\V^r_{\alpha}\}$ of sections of $\T^EE^k\to E^k$, as in Section~\ref{preliminaries},
\[
\X_{\alpha}(A)=\left(A,e_{\alpha}(m),{\rho}^i_{\alpha}\pd{}{x^i}\at{A}\right),
\qquad
\V_{\alpha}^r(A)=\left(A,0_m,\pd{}{y^{\alpha}_r}\at{A}\right),
\qquad m={\tau}_{k,0}(A).
\]
To write more compact expressions we will use the notation $\V^0_{\alpha}=\X_{\alpha}$.

In local coordinates we have the local expressions
\begin{equation}
\Ti=y^{\alpha}_1\X_{\alpha}+\sum_{j=1}^ky^{\alpha}_{j+1}\,\V^j_{\alpha}=\sum_{j=0}^ky^{\alpha}_{j+1}\,\V^j_{\alpha},
\end{equation}
and 
\begin{equation}
d_{\Ti}F={\rho}^i_{\alpha}y^{\alpha}_1\pd{F}{x^i}+\sum_{j=1}^k y^{\alpha}_{j+1}\pd{F}{y^{\alpha}_j}.
\end{equation}

When $E$ is a Lie algebroid, the operator $d_{\Ti}$ acts also on $p$-forms on $E^k$ producing a $p$-form on $E^{k+1}$. For the dual basis $\{\X^{\alpha},\V^{\alpha}_r\}$ of the local basis $\{\X_{\alpha},\V_{\alpha}^r\}$ we have 
\begin{equation}
\label{dT.of.the.basis}
\begin{aligned}
d_{\Ti}\X^{\alpha}&=\V^{\alpha}_1-C^{\alpha}_{{\beta}{\gamma}}y^{\beta}_1\X^{\gamma},&
\\
d_{\Ti}\V^{\alpha}_r&=\V^{\alpha}_{r+1},&r=1,\ldots,k
\end{aligned}
\end{equation}
which follows easily from the definition $d_{\Ti}=d{\circ}i_{\Ti}+i_{\Ti}{\circ}d$. 

\medskip

\paragraph{Jet prolongation of admissible maps}

Consider a second Lie algebroid $\map{{\tau}'}{E'}{M'}$ with anchor ${\rho}'$. A vector bundle map $\map{{\Phi}}{E}{E'}$ is said to be admissible if it maps admissible curves into admissible curves. If $\map{{\varphi}}{M}{M'}$ is the base map, then ${\Phi}$ is admissible if and only if ${\rho}'{\circ}{\Phi}=T{\varphi}{\circ}{\rho}$. 

An admissible map induces a map $\map{{\Phi}^k}{E^k}{E'^k}$ by jet-prolongation as follows: if $a$ is an admissible curve then ${\Phi}{\circ}a$ is admissible and we set ${\Phi}^k([a]^{k-1})=[{\Phi}{\circ}a]^{k-1}$. In other words,  ${\Phi}^k$ is the restriction of the $(k-1)$-jet extension $\map{T^{k-1}{\Phi}}{T^{k-1}E}{T^{k-1}E'}$ of ${\Phi}$ to $E^k$, which takes values in $E'^k$. It is clear from the definition that ${\tau}_{k,l}{\circ}{\Phi}^k={\Phi}^l{\circ}{\tau}_{k,l}$ for $k>l>0$, which is also valid for $k=0$ if we set ${\Phi}^0={\varphi}$.

When we look at $E^k$ as a submanifold of $\T^EE^{k-1}$ we have that ${\Phi}^k$ is the restriction to $E^k$ of the map $\map{\T^{\Phi}{\Phi}^{k-1}}{\T^EE^{k-1}}{\T^{E'}E'^{k-1}}$ given by $(A,b,V)\mapsto({\Phi}^{k-1}(A),{\Phi}(b),T{\Phi}(V))$. In other words, we have $\T^{\Phi}{\Phi}^k{\circ}i_{k,1}=i_{k,1}{\circ}{\Phi}^{k+1}$, which in terms of the operator $\Ti$ reads $\T^{\Phi}{\Phi}^k{\circ}\Ti=\Ti{\circ}{\Phi}^{k+1}$


\section{Variational vector fields and complete lifts}
\label{variational.vector.fields}

As it was explained above, a vector tangent to the $k$-tangent space to a manifold is determined by a 1-parameter family of curves. Accordingly, a vector tangent to $E^k$ is determined by a 1-parameter family ${\alpha}(s,t)$ of admissible curves in $E$, by the same procedure: $[s\mapsto[t\mapsto {\alpha}(s,t)]^{k-1}]^1$ is a vector tangent to $E^k$ at the point $[t\mapsto {\alpha}(0,t)]^{k-1}{\in}E^k$. 

In the calculus of variations on Lie algebroids the families ${\alpha}(s,t)$ of admissible curves are given by morphisms of Lie algebroids $\map{{\phi}}{T{\Real}^2}{E}$ of the form ${\alpha}(s,t)dt+{\beta}(s,t)ds$. The variational vector field $\frac{d}{ds}{\alpha}_s{}^k(t)\big|_{s=0}$ can be determined in terms of ${\sigma}(t)={\beta}(0,t)$ and its derivatives up to order $k$. It is a vector field along $a^k(t)$, where $a(t)={\alpha}(0,t)$, and will be denoted ${\Xi}_a^k{\sigma}(t)$.

Indeed, let us find the local expression of the variational vector field associated to the morphism ${\alpha}dt+{\beta}ds$. In local coordinates, the family ${\alpha}$ is $({\gamma}^i(s,t),{\alpha}^{\mu}(s,t))$ and the family ${\beta}$ is $({\gamma}^i(s,t),{\beta}^{\mu}(s,t))$. The fact that ${\alpha}dt+{\beta}ds$ is a morphism amounts to 
\begin{equation}\label{morphism.local}
\left\{\begin{aligned}
\pd{{\gamma}^i}{t}&={\rho}^i_{\mu}{\alpha}^{\mu}\\
\pd{{\gamma}^i}{s}&={\rho}^i_{\mu}{\beta}^{\mu}\\
\pd{{\alpha}^{\mu}}{s}&=\pd{{\beta}^{\mu}}{t}+C^{\mu}_{{\nu}{\gamma}}{\alpha}^{\nu}{\beta}^{\gamma},
\end{aligned}\right.
\end{equation}
where the local structure functions ${\rho}^i_{\mu}$ and $C^{\mu}_{{\nu}{\gamma}}$ are evaluated at the point ${\gamma}(s,t)$. The curve $a_s{}^k(t)$ in $E^k$ defined by the family ${\alpha}$ is given by
\[
x^i={\gamma}^i(s,t),\ y^{\mu}_1={\alpha}^{\mu}(s,t),\ y^{\mu}_2=\pd{{\alpha}^{\mu}}{t}(s,t),\ldots,\ y^{\mu}_k=\pd{^{k-1}{\alpha}^{\mu}}{t^{k-1}}(s,t),
\]
and the coordinates of ${\Xi}_a^k{\sigma}(t):=\frac{d}{ds}a_s{}^k(t)\at{s=0}$ are
\[
w^i=\pd{{\gamma}^i}{s}(0,t),\ v^{\mu}_1=\pd{{\alpha}^{\mu}}{s}(0,t),\ldots,\ v^{\mu}_k=\pd{^k{\alpha}^{\mu}}{t^{k-1}\partial s}(0,t).
\]
Taking into account the equations~\eqref{morphism.local} we have
\[
w^i={\rho}^i_{\mu}{\beta}^{\mu}(0,t)={\rho}^i_{\mu}{\sigma}^{\mu}(t)
\]
and
\[
v_1^{\mu}=\pd{{\beta}^{\mu}}{t}(0,t)+C^{\mu}_{{\nu}{\gamma}}{\alpha}^{\nu}(0,t){\beta}^{\gamma}(0,t)=\dot{{\sigma}}^{\mu}(t)+C^{\mu}_{{\nu}{\gamma}}a^{\nu}(t){\sigma}^{\gamma}(t),
\]
and hence we have 
\[
v_r^{\mu}=\frac{d^{r-1}v_1^{\mu}}{dt^{r-1}}(t)=\frac{d^{r-1}}{dt^{r-1}}[\dot{{\sigma}}^{\mu}+C^{\mu}_{{\nu}{\gamma}}a^{\nu}{\sigma}^{\gamma}],\qquad r=2,\ldots,k.
\]
It conclusion, we have
\begin{equation}\label{variational.vector.coordinates}
{\Xi}_a^k{\sigma}={\rho}^i_{\alpha}{\sigma}^{\alpha}\pd{}{x^i}+\sum_{r=1}^{k}\frac{d^{r-1}}{dt^{r-1}}[\dot{{\sigma}}^{\alpha}+C^{\alpha}_{{\beta}{\gamma}}a^{\beta}{\sigma}^{\gamma}]\pd{}{y^{\alpha}_r}.
\end{equation}

It follows that ${\Xi}_a^k{\sigma}$ is a differential operator in ${\sigma}$ of order $k$ (depends on $[{\sigma}]^{(k)}$) and a differential operator in $a$ of order $k-1$ (depends on ${a}^k(t)=[h\mapsto a(t+h)]^{k-1}$).

\begin{remark}
In the classical notation of the calculus of variations we have 
\[
{\delta}x^i={\rho}^i_{\alpha}{\sigma}^{\alpha},
\quad
{\delta}y^{\alpha}_1=\dot{{\sigma}}^{\alpha}+C^{\alpha}_{{\beta}{\gamma}}a^{\beta}{\sigma}^{\gamma},
\quad
{\delta}y^{\alpha}_{r}=\frac{d}{dt}{\delta}y^{\alpha}_{r-1},\quad\text{for $r=2,\ldots,k$,}
\]
as it should be expected.
\end{remark}

The variational vector field ${\Xi}_a^k{\sigma}$ can alternatively be defined directly in terms of ${\sigma}$ and $a$, without reference to the morphism ${\phi}$, by jet-prolongation as follows.

\begin{proposition}
Let $a$ be an admissible curve and let ${\sigma}$ be a section of $E$ along ${\tau}{\circ}a$. Consider the associated first order variational vector field $\map{{\Xi}_a{\sigma}}{{\Real}}{TE}$, and its $(k-1)$-jet prolongation $({\Xi}_a{\sigma})^{(k-1)}(t){\in}T^{k-1}TE$. Then, the vector field ${\chi}_{k-1}{\circ}({\Xi}_a{\sigma})^{(k-1)}$ is tangent to $E^k$ and coincides with ${\Xi}_a^k{\sigma}$.
\end{proposition}
\begin{proof}
It is enough to prove the proposition for $t=0$; for $t=t_0{\neq}0$ we can just consider the curves $\bar{a}(t)=a(t_0+t)$ and $\bar{{\sigma}}(t)={\sigma}(t_0+t)$.

We take ${\alpha}(s,t)$ a family of admissible curves such that ${\alpha}(0,t)=a(t)$ and find a complementary ${\beta}(s,t)$ such that $\map{{\alpha}dt+{\beta}ds}{T{\Real}^2}{E}$ is a morphism of Lie algebroids with ${\beta}(0,t)={\sigma}(t)$. On one hand, by construction, we have that $[s\mapsto [t\mapsto {\alpha}(s,t)]^{k-1}]^1$ takes values in $TE^k$. On the other hand, $\pd{{\alpha}}{s}(0,t)={\Xi}_a{\sigma}(t)$, from where
\begin{align*}
\frac{d}{ds}{\alpha}_s{}^k(0)\at{s=0}
&=[s\mapsto [t\mapsto {\alpha}(s,t)]^{k-1}]^1\\
&={\chi}_{k-1}([t\mapsto [s\mapsto {\alpha}(s,t)]^1]^{k-1})\\
&={\chi}_{k-1}([t\mapsto {\Xi}_a{\sigma}(t)]^{k-1})\\
&=\{{\chi}_{k-1}{\circ}({\Xi}_a{\sigma})^{(k-1)}\}(0),
\end{align*}
and the result follows.
\end{proof}

As an immediate consequence of the above proposition we have
\begin{equation}
T{\tau}_{k,l}{\circ}{\Xi}_a^k{\sigma}={\Xi}^l_a{\sigma},\qquad\qquad\text{for $k>l>0$.}
\end{equation}
It follows that ${\Xi}^k_a{\sigma}$ projects to the vector field ${\rho}({\sigma})$, i.e.\ $T{\tau}_{k,0}{\circ}{\Xi}_a^k{\sigma}={\rho}({\sigma})$, and hence it make sense the following definition.
\begin{definition}
Let $a$ be an admissible curve and ${\sigma}$ be a section along ${\gamma}={\tau}{\circ}a$. The section ${\sigma}_a^k$ of $\T^EE^k$ along $a^k$ given by 
\[
{\sigma}_a^k(t)=\bigl(a^k(t),{\sigma}(t),{\Xi}_a^k{\sigma}(t)\bigr)
\]
is said to be the $k$th-order complete lift of ${\sigma}$ with respect to $a$.
\end{definition}

It is easy to see that every element in $\T^EE^k$ is of the form ${\sigma}^k_a(0)$ for some section~${\sigma}$. 

\medskip

\paragraph{Complete lift of a section of $E$}
The concept of variational vector field is related to the concept of complete lift of a section of $E$. 

Let ${\eta}$ be  a section of $E$ and $({\Phi}_s,{\varphi}_s)$ its flow~\cite{Rui,VCLA}. The map $\map{{\Phi}_s{^k}}{E^k}{E^k}$ is a flow in $E^k$ which defines a vector field $X_{\eta}^k{\in}TE^k$. This vector field is ${\tau}_{k,l}$-projectable over $X^l_{\eta}$ for $k>l>0$, and ${\tau}_{k,0}$-projectable over ${\rho}({\eta})$. Therefore it allows to define a section ${\eta}^k{\in}\sec{\T^EE^k}$ by 
\[
{\eta}^k(A)=(A,{\eta}({\tau}_{k,0}(A)),X^k_{\eta}(A)),\qquad A{\in}E^k.
\]
The section ${\eta}^k$ will be called the $k$th-order \emph{complete lift} of the section ${\eta}$. For $k=1$ the section ${\eta}^1$ coincides with the complete lift ${\eta}^c$ of the section ${\eta}$ as defined in~\cite{LMLA}. From the definition it is clear that $\T{\tau}_{k,l}{\circ}{\eta}^k={\eta}^l{\circ}{\tau}_{k,l}$, for $k>l{\geq}0$, where ${\eta}^0={\eta}$ should be understood for $l=0$.

\begin{proposition}
For any section ${\eta}{\in}\sec{E}$ we have the following property
\begin{equation}\label{complete.lift.commutes.with.T}
d_{\Ti}{\circ}i_{{\eta}^k}=i_{{\eta}^{k+1}}{\circ}d_{\Ti}.
\end{equation}
Conversely, if $Z$ is a section of $\T^EE^{k+1}$ which projects to ${\eta}{\in}\sec{E}$ and satisfies $d_{\Ti}{\circ}i_{{\eta}^k}=i_Z{\circ}d_{\Ti}$, then $Z={\eta}^{k+1}$.
\end{proposition}
\begin{proof}
It is enough to prove the proposition over basic forms ${\theta}{\in}\sec{E^*}$, and over forms of the type $d^r(dF)$ for $F{\in}\cinfty{E}$ and $r=0,\ldots,k-1$ (we have omitted the pullbacks for simplicity), since they generate the set of sections of $(\T^EE^k)^*$.

For ${\theta}{\in}\sec{E^*}$, in local coordinates if ${\theta}={\theta}_{\alpha}e^{\alpha}$ and ${\eta}={\eta}^{\alpha}e_{\alpha}$, we have 
\[
d_{\Ti}{\theta}=(d_{\Ti}{\theta}_{\alpha})\X^{\alpha}+{\theta}_{\alpha}(\V^{\alpha}-C^{\alpha}_{{\beta}{\gamma}}y^{\beta}\X^{\gamma})
\]
and
\[
{\eta}^1={\eta}^{\alpha}\X_{\alpha}+(d_{\Ti}{\eta}^{\alpha}+C^{\alpha}_{{\beta}{\gamma}}y^{\beta}{\eta}^{\gamma})\V_{\alpha},
\]
from where 
\[
\pai{d_{\Ti}{\theta}}{{\eta}^1}=(d_{\Ti{\theta}_{\alpha}}){\eta}^{\alpha}+{\theta}_{\alpha}(d_{\Ti}{\eta}^{\alpha})=d_{\Ti}\pai{{\theta}}{{\eta}}.
\]

We will prove that for any function $F{\in}\cinfty{E}$ and $r=0,\ldots,k$ we have 
\[
d_{\Ti}\pai{d^r_{\Ti}dF}{{\eta}^{r+1}}=\pai{d^{r+1}_{\Ti}dF}{{\eta}^{r+2}},
\]
which can be equivalently stated in the form 
\[
d_{\Ti}d_{{\eta}^{r+1}}d^r_{\Ti}F=d_{{\eta}^{r+2}}d^{r+1}_{\Ti}F.
\]
Consider an arbitrary point $A=[a]^{r+1}{\in}E^{r+2}$ and the flow ${\Phi}_s$ of the section ${\eta}$. Then the left hand side evaluated at $A$ is equal to
\begin{align*}
(d_{\Ti}d_{{\eta}^{r+1}}d^r_{\Ti}F)(A)
&=\frac{d}{dt}[(d_{{\eta}^{r+1}}d^r_{\Ti}F){\circ}a^{r+1}](0)\\
&=\frac{d}{dt}\frac{d}{ds}[d^r_{\Ti}F{\circ}{\Phi}^{r+1}_s{\circ}a^{r+1}](t)\at{s=0}\at{t=0}\\
&=\frac{d}{dt}\frac{d}{ds}[d^r_{\Ti}F{\circ}({\Phi}_s{\circ}a)^{r+1}](t)\at{s=0}\at{t=0}\\
&=\frac{d}{dt}\frac{d}{ds}\Bigl[\frac{d^r}{dt^r}(F{\circ}{\Phi}_s{\circ}a)(t)\Bigr]\at{t=0}\at{s=0}\\
&=\frac{d}{ds}\frac{d^{r+1}}{dt^{r+1}}(F{\circ}{\Phi}_s{\circ}a)(t)\at{t=0}\at{s=0}
\end{align*}
and the right hand side evaluated at $A$ is
\begin{align*}
(d_{{\eta}^{r+2}}d^{r+1}_{\Ti}F)(A)
&=\frac{d}{ds}(d^{r+1}_{\Ti}F{\circ}{\Phi}^{r+2}_s)(A)\at{s=0}\\
&=\frac{d}{ds}(d^{r+1}_{\Ti}F{\circ}{\Phi}^{r+2}_s{\circ}a^{r+2})(0)\at{s=0}\\
&=\frac{d}{ds}(d^{r+1}_{\Ti}F{\circ}({\Phi}_s{\circ}a)^{r+2})(0)\at{s=0}\\
&=\frac{d}{ds}\frac{d^{r+1}}{dt^{r+1}}(F{\circ}{\Phi}_s{\circ}a)(t)\at{t=0}\at{s=0}
\end{align*}
and both expressions are equal.

Conversely, let $Z$ be a section of $\T^EE^{k+1}$ such that $\T{\tau}_{k+1,0}{\circ}Z={\eta}{\circ}{\tau}_{k+1,0}$ and $i_Zd_{\Ti}{\lambda}=d_{\Ti}i_{{\eta}^k}{\lambda}$ for every section ${\lambda}$ of $(\T^EE^k)^*$. In particular for ${\lambda}=d^{r-1}_{\Ti}{\theta}$ with $r=1,\ldots,k$ we have $i_Z(d^r_{\Ti}{\theta})=d_{\Ti}i_{{\eta}^k}d^{r-1}_{\Ti}{\theta}=i_{{\eta}^{k+1}}(d^r_{\Ti}{\theta})$, so that $i_{Z-{\eta}^{k+1}}d^{r}_{\Ti}{\theta}=0$, for $r=1,\ldots,k$. Moreover, since $Z$ projects to ${\eta}$ it follows that this last relation holds also for $r=0$. Therefore $Z-{\eta}^{k+1}=0$.
\end{proof}

Notice that as a consequence of the above property we have that 
\begin{equation}
d_{\Ti}{\circ}d_{{\eta}^k}=d_{{\eta}^{k+1}}{\circ}d_{\Ti},
\end{equation}
which follows from $d_{{\eta}^r}=d{\circ}i_{{\eta}^r}+i_{{\eta}^r}{\circ}d$ and $d{\circ}d_{\Ti}=d_{\Ti}{\circ}d$.

\medskip

\paragraph{Relation of complete lifts with variational vector fields}

The relation between complete lifts of sections and variational vector fields is as follows. 

Let $a$ be an admissible curve over ${\gamma}$. Consider a section ${\eta}$ of $E$ and define the section ${\sigma}$ along ${\gamma}$ by ${\sigma}(t)={\eta}({\gamma}(t))$. Consider on one hand the complete lift ${\sigma}^k_a$ of ${\sigma}$ with respect to $a$, and on the other the complete lift ${\eta}^k$ of the section ${\eta}$. Then 
\begin{equation}\label{relation.variational.complete}
{\sigma}^k_a={\eta}^k{\circ}a^k.
\end{equation}
Indeed, it is shown in~\cite{VCLA} that the section ${\eta}$ and the admissible curve $a$ define a morphism ${\phi}={\alpha}(s,t)dt+{\beta}(s,t)ds={\Phi}_s(a(t))dt+{\eta}({\varphi}_s({\gamma}(t)))ds$. At $s=0$ we have ${\alpha}(0,t)=a(t)$, ${\beta}(0,t)={\eta}({\gamma}(t))={\sigma}(t)$ and the vector field $\frac{d}{ds}a_s^k\at{s=0}(t)={\Xi}^k_a{\sigma}$ is equal to $X^k_{\eta}(a^k(t))$. Therefore 
\[
{\sigma}^k_a(t)=(a^k(t),{\sigma}(t),{\Xi}_a^k{\sigma}(t))=\bigl(a^k(t),{\eta}({\gamma}(t)),X^k_{\eta}(a^k(t)\bigr)={\eta}^k(a^k(t)),
\]
form where ${\sigma}_a^k(t)={\eta}^k(a^k(t))$ follows.

Conversely, given ${\sigma}(t)$ along ${\gamma}(t)$ we can find a time-dependent section ${\eta}$ such that ${\eta}(t,{\gamma}(t))={\sigma}(t)$. Using the construction for the time-dependent section as in~\cite{VCLA} we also have ${\sigma}^k(t)={\eta}^k(t,a^k(t))$.

It follows that if the coordinate expression of ${\eta}$ is ${\eta}={\eta}^{\alpha}e_{\alpha}$ then the coordinate expression of the complete lift is 
\begin{equation}
\label{complete.lift.section.coordinates}
{\eta}^k={\eta}^{\alpha}\X_{\alpha}+\sum_{r=1}^kd^{r-1}_{\Ti}[d_{\Ti}{\eta}^{\alpha}+C^{\alpha}_{{\beta}{\gamma}}y_1^{\beta}{\eta}^{\gamma}]\,\V^r_{\alpha}.
\end{equation}
Simmilarly
\begin{equation}
\label{complete.lift.sections.along.coordinates}
{\sigma}^k_a(t)={\sigma}^{\alpha}(t)\X_{\alpha}(a^k(t))+\sum_{r=1}^k\frac{d^{r-1}}{dt^{r-1}}[\dot{{\sigma}}^{\alpha}+C^{\alpha}_{{\beta}{\gamma}}a^{\beta}{\sigma}^{\gamma}](t)\,\V^r_{\alpha}(a^k(t)).
\end{equation}
where ${\sigma}={\sigma}^{\alpha}(t)e_{\alpha}({\gamma}(t))$.

\section{The vertical endomorphism and the variational operator}
\label{vertical.endomorphism}

Working with forms and their time derivatives, it is useful to have an operator which formally resembles integration, that is, it is a kind of inverse of $d_{\Ti}$. This is the role of the vertical endomorphism, which is an endomorphism of the vector bundle $\T^EE^k$. We introduce the vertical endomorphism as an auxiliary tool to define two differential operators, the variational operator and the Cartan operator, which will be fundamental in the development of the calculus of variations. 


\subsection{Vertical endomorphism}
We will define the vertical endomorphism $S$ by its action on sections of the dual bundle. We will make no notational distinction between $S$ and its dual $S^*$, both will be written $S$. 

The set of sections of $(\T^EE^k)^*$ is generated (with coefficients in $\cinfty{E^k}$) by sections of the form $(\T{\tau}_{k,r+1})^\star d_{\Ti}^{r}\theta$ for sections $\theta{\in}\sec{E^*}$ and $r=0,\ldots,k$. In particular, if  $\{e^{\alpha}\}$ is a local basis of $\sec{E^*}$, then a local basis of $\sec{(\T^EE^k)^*}$ is $\{(\T{\tau}_{k,r+1})^\star d^{r}_{\Ti}e^{\alpha}\}$ for $r=0,\ldots,k$. 
Therefore, it is sufficient to give the action of $S$ over sections of such a form. To simplify the notation, in what follows we will omit the pullbacks

We define the \emph{vertical endomorphism} as the $(1,1)$ tensor field $\map{S}{\T^EE^k}{\T^EE^k}$ which satisfies
\begin{equation}\label{definition.S}
S(d_{\Ti}^r\theta)=r d_{\Ti}^{r-1}{\theta},
\end{equation}
where $S({\theta})=0$ should be understood for $r=0$. Notice also that for $r=1$ we have $S(d_{\Ti}{\theta})={\theta}$.

We have to check the consistency of the above formula when taking linear combinations of sections. Obviously there is no problem when taking sums, so that we have to check consistency when multiplying by functions on the base. If we take $f{\theta}$, with $f{\in}\cinfty{M}$, then
\[
d^r_{\Ti}(f{\theta})
=\sum_{j=0}^r{r\choose j}(d_{\Ti}^{r-j}f)(d_{\Ti}^j{\theta}),
\]
so that 
\begin{align*}
S(d^r_{\Ti}(f{\theta}))
&=\sum_{j=0}^r{r\choose j}(d_{\Ti}^{r-j}f)S(d_{\Ti}^j{\theta}) \\
&=\sum_{j=0}^r{r\choose j}j(d_{\Ti}^{r-j}f)d_{\Ti}^{j-1}{\theta}\\
&=\sum_{j=1}^r r{r-1\choose j-1}(d_{\Ti}^{r-j}f)d_{\Ti}^{j-1}{\theta}\\
&=rd^{r-1}_{\Ti}(f{\theta}),
\end{align*}
consistently with the given definition.

\medskip

A section ${\lambda}$ of $(\T^EE^k)^*$ is said to be \emph{${\tau}_{k,r}$-semibasic} if it vanishes over elements in the kernel of $\T{\tau}_{k,r}$. In other words, ${\lambda}$ is ${\tau}_{k,r}$-semibasic if there exists a map $\map{\bar{{\lambda}}}{\T^EE^k}{(\T^EE^r)^*}$ such that ${\lambda}(Z)=\pai{\bar{{\lambda}}}{\T{\tau}_{k,r}(Z)}$ for every $Z{\in}\T^EE^r$. In particular, ${\lambda}$ is ${\tau}_{k,0}$-semibasic if there exists a map $\map{\bar{{\lambda}}}{\T^EE^k}{E^*}$ such that ${\eta}(Z)=\pai{\bar{{\lambda}}}{\T{\tau}_{k,0}(Z)}$ for every $Z{\in}\T^EE^r$.  It is clear that the space of ${\tau}_{k,l}$-semibasic forms is generated (over $\cinfty{E^k}$) by the family of sections of the form $d_{\Ti}^r{\theta}$ for $r=0,\ldots,l$. 

\begin{proposition}\label{algebraic.properties.S}
A section ${\lambda}\in\sec{(\T^EE^k)^*}$ is ${\tau}_{k,l}$-semibasic if and only if $S^{l+1}({\eta})=0$. The image of $S^{l}$ is the set of ${\tau}_{k,k-l}$-semibasic forms. In particular, we have that $S^{k+1}=0$.
\end{proposition}
\begin{proof}
For $r=0,\ldots,k$, a simple calculation shows that 
\[
S^p(d^r_{\Ti}{\theta})
=\begin{cases}
0&\text{if $p>r$,}\\
\frac{p!}{(p-r)!}d^{r-p}_{\Ti}{\theta}&\text{if $p{\leq}r$.}
\end{cases}
\]
It follows that $S^p({\lambda})=0$ if and only ${\lambda}$ can be written as a linear combination (with coefficients in $\cinfty{E^k}$) of forms of the type $d^r_{\Ti}{\theta}$ with $r=0,\ldots,p-1$. Therefore $S^p({\lambda})=0$ if and only if ${\lambda}$ is ${\tau}_{k,p-1}$-semibasic. Moreover, from this expression we have that the image of $S^p$ are ${\tau}_{k,k-p}$-semibasic forms. A simple argument shows that every ${\tau}_{k,k-p}$-semibasic form is in the image of $S^p$.
\end{proof}

\begin{proposition}\label{S.dT-dT.S=1}
For any section ${\lambda}{\in}\sec{(\T^EE^{k-1})^*}$ we have 
\begin{equation}
S(d_{\Ti}{\lambda})-d_{\Ti}(S{\lambda})={\lambda}.
\end{equation}
\end{proposition}
\begin{proof}
It is sufficient to prove the proposition for a section ${\lambda}$ of the form ${\lambda}=d_{\Ti}^r{\theta}$ for $r=0,\ldots,k-1$. Using the definition of $S$ we have
\begin{align*}
S(d_{\Ti}{\lambda})-d_{\Ti}(S{\lambda})
&=S(d_{\Ti}^{r+1}{\theta})-d_{\Ti}(S(d_{\Ti}^r{\theta}))\\
&=(r+1)d_{\Ti}^{r}{\theta}-d_{\Ti}(rd_{\Ti}^r{\theta})
=d_{\Ti}^r{\theta}
={\lambda},
\end{align*}
which proves the statement.
\end{proof}

\paragraph{Coordinate expression}
Even though it is not needed in this paper, we will find the local expression of the vertical endomorphism. We take a local basis $\{e_{\alpha}\}$ of sections of $E$ and the dual basis $\{e^{\alpha}\}$ of sections of $E^*$. Then we have a local basis $\{\X_{\alpha},\V^r_{\alpha}\}$ of sections of $\T^EE^k$ and the dual basis $\{\X^{\alpha},\V^{\alpha}_r\}$ of sections of $(\T^EE^k)^*$. We have to find the image by $S$ of such sections. From the definition it follows that $S(\X^{\alpha})=0$. To find the expression of $S(\V^{\alpha}_1)$ we use equation~\eqref{dT.of.the.basis}, $d_{\Ti}\X^{\alpha}=\V^{\alpha}_1-C^{\alpha}_{{\beta}{\gamma}}y^{\beta}_1\X^{\gamma}$, and thus
\[
S(\V^{\alpha}_1)=S(d_{\Ti}\X^{\alpha}+C^{\alpha}_{{\beta}{\gamma}}y^{\beta}_1\X^{\gamma})=S(d_{\Ti}\X^{\alpha})=d_{\Ti}S(\X^{\alpha})+\X^{\alpha}=\X^{\alpha}.
\]
Finally, in order to find $S(\V^{\alpha}_r)$ for $r{\geq}2$ we recall equation~\eqref{dT.of.the.basis}, $\V^{\alpha}_r=d_{\Ti}\V^{\alpha}_{r-1}$, and using Proposition~\ref{S.dT-dT.S=1} we have
\[
S(\V^{\alpha}_r)=S(d_{\Ti}\V^{\alpha}_{r-1})=d_{\Ti}S(\V^{\alpha}_{r-1})+\V^{\alpha}_{r-1}
\]
for $r{\geq}2$. It easily follows by induction that 
\[
S(\V^{\alpha}_r)=(r-1)\V^{\alpha}_{r-1}+d_{\Ti}^{r-1}S(\V^{\alpha}_1),\qquad r{\geq}2.
\]
Therefore we found
\begin{equation}\label{local.S.1}
S=\V_{\alpha}^1\otimes\X^{\alpha}+\sum_{r=2}^k\left[(r-1)\V^r_{\alpha}\otimes\V^{\alpha}_{r-1}+\V^r_{\alpha}\otimes d^r_{\Ti}\X^{\alpha}\right].
\end{equation}

\begin{remark}
A more explicit local expression for $S$ can be obtained as follows. The value of $S(\V^{\alpha}_r)$ can be written in the form 
\[
S(\V^{\alpha}_r)=r\V^{\alpha}_{r-1}+\bigl[d_{\Ti}^{r-1}S(\V^{\alpha}_1)-\V^{\alpha}_{r-1}\bigr]\qquad r{\geq}2.
\]
The expression $d^r_{\Ti}\X^{\alpha}-\V^{\alpha}_r$ is a $(r-1)$th total time derivative of a semibasic 1-form 
\[
d^r_{\Ti}\X^{\alpha}-\V^{\alpha}_r=-d_{\Ti}^{r-1}(C^{\alpha}_{{\beta}{\gamma}}y^{\beta}_1\X^{\gamma}),
\]
which follows from $d_{\Ti}\X^{\alpha}=\V^{\alpha}_1-C^{\alpha}_{{\beta}{\gamma}}y^{\beta}_1\X^{\gamma}$. Therefore we can also write
\begin{equation}\label{local.S.2}
S=\V_{\alpha}^1\otimes\X^{\alpha}+\sum_{r=2}^k\left[r\V^r_{\alpha}\otimes\V^{\alpha}_{r-1}-\V^r_{\alpha}\otimes d^{r-1}_{\Ti}(C^{\alpha}_{{\beta}{\gamma}}y^{\beta}_1\X^{\gamma})\right].
\end{equation}
In more compact way
\begin{equation}\label{local.S.compact}
S=\sum_{r=1}^k r\,\V^r_{\alpha}\otimes\V^{\alpha}_{r-1}-\sum_{r=2}^k\V^r_{\alpha}\otimes d^{r-1}_{\Ti}(C^{\alpha}_{{\beta}{\gamma}}y^{\beta}_1\X^{\gamma}).
\end{equation}
\end{remark}


\subsection{Variational and Cartan operators}
\label{variational.and.Cartan}

Having in mind the procedure of integration by parts that will be needed in the calculus of variations, we can define two differential operators (see~\cite{sections} for the standard case): the \emph{variational operator} $\E$,  mapping sections of $(\T^EE^k)^*$ into sections of $(\T^EE^{2k})^*$, given by 
\begin{equation}
\label{variational.operator}
\E({\lambda})
={\lambda}-d_{\Ti}(S({\lambda}))+\frac{1}{2!}d_{\Ti}^2(S^2({\lambda}))+\ldots+(-1)^{k}\frac{1}{k!}d_{\Ti}^k(S^k({\lambda})),
\end{equation}
and the \emph{Cartan operator} $\S$ which maps sections of $(\T^EE^k)^*$ into sections of $(\T^EE^{2k-1})^*$ given by 
\begin{equation}
\label{Cartan.operator}
\S({\lambda})=S({\lambda})-\frac{1}{2!}d_{\Ti}(S^2({\lambda}))+\ldots+(-1)^{k-1}\frac{1}{k!}d_{\Ti}^{k-1}(S^k({\lambda})),
\end{equation}
for ${\lambda}\in\sec{(\T^EE^k)^*}$. From the definition of $\E$ and $\S$ it is clear that 
\begin{equation}\label{relation.variational.Cartan}
\E({\lambda})={\lambda}-d_{\Ti}(\S({\lambda})).
\end{equation}
Moreover, using the Proposition~\ref{S.dT-dT.S=1}, a long but straightforward calculation shows that $S(\E({\lambda}))=0$, so that $\E({\lambda})$ is ${\tau}_{2k,0}$-semibasic, and also $S^k(\S({\lambda}))=0$, so that $\S({\lambda})$ is ${\tau}_{2k-1,k-1}$-semibasic. Moreover $\S(d_{\Ti}{\lambda})={\lambda}$.

\begin{remark}
In a more systematic way, one can proceed as in~\cite{sections}, where a family of operators $\mathbb{D}_r$ was introduced to study complete lifts of vector fields and other related properties. For $s=0,\ldots,k$  we define the operator $\mathbb{D}_r$ mapping sections of $(\T^EE^k)^*$ into sections of $(\T^EE^{2k-r})^*$  be means of 
\[
\mathbb{D}_r({\lambda})=\sum_{j=r}^k (-1)^{j+r}\frac{1}{j!}d_{\Ti}^{j-r}(S^j({\lambda})),
\]
for ${\lambda}\in\sec{(\T^EE^k)^*}$. Following the arguments in~\cite{sections}, it is easy to see that the image of $\mathbb{D}_r$ are $\tau_{2k-r,k-r}$-semibasic sections, and that we have the relations 
\begin{gather}
\mathbb{D}_{r-1}=\frac{1}{r!}S^{r-1}-d_{\Ti}{\circ}\mathbb{D}_r
\label{rel.2}\\
S\circ\mathbb{D}_r=r\mathbb{D}_{r+1}.
\label{rel.1}
\end{gather}
From the second relation~\eqref{rel.1} it follows that only $\mathbb{D}_0$ and $\mathbb{D}_1$ are relevant, while the others can be determined inductively by $\mathbb{D}_{r+1}=\frac{1}{r}S\circ\mathbb{D}_r$ for $r{\geq}1$. The variational operator is $\mathbb{E}=\mathbb{D}_0$ and the Cartan operator is $\mathbb{S}=\mathbb{D}_1$. Also from~\eqref{rel.1} for $r=0$ we have $S\circ\mathbb{D}_0=0$ so that $\E({\lambda})$ is ${\tau}_{2k,0}$-semibasic. Applying $S^{k-1}$ to~\eqref{rel.1} for $r=1$ we get 
\[
S^k(\mathbb{D}_1({\lambda}))=(k-1)!S(\mathbb{D}_k({\lambda}))=\frac{1}{k}S^{k+1}({\lambda})=0,
\]
so that $\S({\lambda})$ is ${\tau}_{2k-1,k-1}$-semibasic.
\end{remark}

\section{Variational calculus}
\label{variational.description}
Let $J=[t_0,t_1]{\subset}{\Real}$ and fix two points $A_0{\in}E^{k-1}$ and $A_1{\in}E^{k-1}$. Given a Lagrangian function $L\in\cinfty{E^k}$ we consider the action functional\footnote{In what follows the symbol  $S$ stands for the action functional. No confusion with the vertical endomorphism is possible since it will not be used explicitly. We will only need the operators $\mathbb{E}$ and $\mathbb{S}$.}
\begin{equation}
S(a)=\int_{t_0}^{t_1}L(a^k(t))\,dt
\end{equation}
restricted to admissible curves $a$ in $E$ such that $a^{k-1}(t_0)=A_0$ and $a^{k-1}(t_1)=A_1$. We look for critical points of such a functional. 

Of course, to speak about critical points, local maxima or local minima, has no meaning if we do not give explicitly a differential manifold structure to the set of such curves. In the case $k=1$ it was shown in~\cite{VCLA} that the relevant Banach manifold structure is the foliated one $\PJ$, where each connected component (a leaf) is an $E$-homotopy class of admissible curves. Therefore, in the higher-order case we will use the same structure and we impose the additional conditions coming from the boundary conditions. In an alternative but equivalent way, we can restrict $S$ to an $E$-homotopy leaf (which is a Banach submanifold) and we use differential calculus to find the critical points. It follows that finite variations are those defined by $E$-homotopies with the given boundary conditions.

We denote by $m_0,m_1{\in}M$ the base points $m_0={\tau}_{k-1,0}(A_0)$ and $m_1={\tau}_{k-1,0}(A_1)$. From the results in~\cite{VCLA}, the following result is easy to prove.

\begin{theorem}
The set 
\begin{equation}
\PJ_{A_0}^{A_1}=\set{a{\in}\PJ}{\text{$a$ is $C^k$ and $a^{k-1}(t_0)=A_0,\ a^{k-1}(t_1)=A_1$}}
\end{equation}
is a Banach submanifold of $\PJ_{m_0}^{m_1}$. The tangent space to $\PJ_{A_0}^{A_1}$ at a point $a{\in}\PJ_{A_0}^{A_1}$ is 
\begin{equation}
T_a\PJ_{A_0}^{A_1}=\set{{\Xi}^k_a{\sigma}}{\text{${\sigma}$ is $C^k$ and ${\sigma}^{(k-1)}(t_0)=0,\ {\sigma}^{(k-1)}(t_1)=0$}}.
\end{equation}
\end{theorem}

\begin{remark}
If $a{\in}\PJ_{m_0}^{m_1}$ is an admissible curve, the connected component of $a$ in $\PJ_{m_0}^{m_1}$ is the equivalence class $H_a$ of admissible curves in $E$ which are $E$-homotopic to $a$. Similarly, if $a{\in}\PJ_{A_0}^{A_1}$, the connected component of $a$ in $\PJ_{A_0}^{A_1}$ is the equivalence class 
$H_a'{\subset}H_a$ of admissible curves in $E$ which are $E$-homotopic to $a$ and have higher-order contact with $a$ at the endpoints.
\end{remark}

It follows that the infinitesimal variations to be used are of the form ${\Xi}_a^k{\sigma}$ with $[{\sigma}]^{k-1}(t_i)=0$, $i=0,1$. Using the classical notation of the calculus of variations
\[
{\delta}x^j={\rho}^j_{\alpha}{\sigma}^{\alpha},\qquad
{\delta}y^{\alpha}_1=\dot{{\sigma}}^{\alpha}+C^{\alpha}_{{\beta}{\gamma}}y^{\beta}_1{\sigma}^{\gamma},\qquad
{\delta}y^{\alpha}_r=\frac{d^{r-1}}{dt^{r-1}}{\delta}y^{\alpha}_1,\quad r=2,\ldots,k,
\]
with $\frac{d^r{\sigma}^{\alpha}}{dt^r}(t_i)=0$ for $r=0,\ldots,k-1$, $i=0,1$.

Given an admissible curve $a$ we take a curve ${\alpha}_s$ on $\PJ_{A_0}^{A_1}$ with ${\alpha}_0=a$, and the corresponding $E$-homotopy ${\alpha}(s,t)dt+{\beta}(s,t)ds$ where ${\alpha}_s(t)={\alpha}(s,t)$. Taking the derivative at $s=0$,
\[
\frac{d}{ds} S({\alpha}_s)\at{s=0}=\int_{t_0}^{t_1}\frac{d}{ds}L({\alpha}_s{}^k(t))\at{s=0}\,dt
=\int_{t_0}^{t_1}\pai{\d L(a^k(t))}{\frac{d}{ds}{\alpha}_s{}^k(t)\at{s=0}}\,dt.
\]
Defining ${\sigma}(t)={\beta}(0,t)$ we have that $\frac{d}{ds}{\alpha}_s{}^k(t)\at{s=0}={\Xi}_a^k{\sigma}(t)$ and hence 
\begin{equation}
\pai{\d S(a)}{{\Xi}_a{\sigma}}
=\frac{d}{ds} S({\alpha}_s)\at{s=0}
=\int_{t_0}^{t_1}\pai{\d L(a^k(t))}{{\Xi}_a^k{\sigma}(t)}\,dt.
\end{equation}
In the above expressions $\d$ stands for the standard exterior differential on a manifold. Using the exterior differential $d$ of the algebroid $\T^EE^k$ and taking into account that ${\sigma}^k_a(t)=(a^k(t),{\sigma}(t),{\Xi}_a^k{\sigma}(t))$ we have
\begin{equation}
\mathrm{d}S(a)({\Xi}_a{\sigma})=\int_{t_0}^{t_1}\pai{dL(a^k(t))}{{\sigma}_a^{k}(t)}\,dt.
\end{equation}

Let us consider the sections $\E(dL)$ and ${\theta}_L=\S(dL)$, which are related by $\E(dL)=dL-d_{\Ti}{\theta}_L$. We have 
\begin{align*}
\pai{dL(a^k(t))}{{\sigma}_a^k(t)}
&=\pai{\E(dL)(a^{2k}(t))}{{\sigma}_a^{2k}(t)}+\pai{(d_{\Ti}{\theta}_L)(a^{2k}(t))}{{\sigma}_a^{2k}(t)}\\
&=\pai{\E(dL)(a^{2k}(t))}{{\sigma}_a^{2k}(t)}+\frac{d}{dt}\pai{{\theta}_L(a^{2k-1}(t))}{{\sigma}_a^{2k-1}(t)}.
\end{align*}
Taking into account that $\E(dL)$ is ${\tau}_{k,0}$-semibasic we will consider the associated map $\map{{\delta}L}{E^{2k}}{E^*}$. Similarly, taking into account that ${\theta}_L=\S(dL)$ is ${\tau}_{2k-1,k-1}$-semibasic we will consider the associated map $\map{\FL}{E^{2k-1}}{(\T^EE^{k-1})^*}$. We then have
\[
\pai{dL(a^k(t)}{{\sigma}_a^k(t)}
=\pai{{\delta}L(a^{2k}(t))}{{\sigma}(t)}+\frac{d}{dt}\pai{\FL(a^{2k-1}(t))}{{\sigma}_a^{k-1}(t)}.
\]
Inserting this into the variation of the action we get
\begin{align*}
\mathrm{d}S(a)({\Xi}_a{\sigma})
&=\int_{t_0}^{t_1}\left[\pai{{\delta}L(a^{2k}(t))}{{\sigma}(t)}+\frac{d}{dt}\pai{\FL(a^{2k-1}(t))}{{\sigma}_a^{k-1}(t)}\right]\,dt\\
&=\int_{t_0}^{t_1}\pai{{\delta}L(a^{2k}(t))}{{\sigma}(t)}\,dt+\pai{\FL(a^{2k-1}(t))}{{\sigma}_a^{k-1}(t)}\Big|_{t_0}^{t_1}.
\end{align*}
After  imposing the boundary conditions we arrive to 
\begin{equation}
\mathrm{d}S(a)({\Xi}_a{\sigma})=\int_{t_0}^{t_1}\pai{{\delta}L(a^{2k}(t)}{{\sigma}(t)}\,dt.
\end{equation}
Since ${\sigma}$ is arbitrary, from the fundamental lemma of the Calculus of Variations, we find that a curve $a$ is a critical point of the action if and only if it satisfies ${\delta}L(a^{2k}(t))=0$.  We have proved the following result.

\begin{theorem}
An admissible curve $a{\in}\PJ_{A_0}^{A_1}$ is a critical point of the action functional $\map{S}{\PJ_{A_0}^{A_1}}{{\Real}}$ if and only if it satisfies the \emph{Euler-Lagrange equations} ${\delta}L(a^{2k}(t))=0$ for all $t{\in}J$.
\end{theorem}

In order to find the local expression of the Euler-Lagrange equations we just calculate the variations of the Lagrangian. To simplify the notation we will use the symbol ${\delta}y^{\alpha}_r$ for the expression ${\delta}y^{\alpha}_r=\frac{d^{r-1}}{dt^{r-1}}[\dot{{\sigma}}^{\alpha}+C^{\alpha}_{{\beta}{\gamma}}y^{\beta}_1{\sigma}^{\gamma}]$, so that
\[
\pai{dL}{{\sigma}^k_a}
={\rho}^i_{\alpha}{\sigma}^{\alpha}\pd{L}{x^i}+\sum_{r=1}^k{\delta}y^{\alpha}_r\pd{L}{y^{\alpha}_r}.
\]
A straightforward calculation shows that, for $r{\geq}2$,  
\[
\pd{L}{y^{\alpha}_r}{\delta}y^{\alpha}_r=(-1)^{r+1}\frac{d^{r-1}}{dt^{r-1}}\left(\pd{L}{y^{\alpha}_r}\right){\delta}y^{\alpha}_1
+\frac{d}{dt}
 \left[\sum_{j=0}^{r-2}(-1)^j\frac{d^{j}}{dt^{j}}\left(\pd{L}{y^{\alpha}_r}\right){\delta}y^{\alpha}_{r-j-1}\right],
\]
and hence 
\begin{align*}
\pai{dL}{{\sigma}^k_a}
&={\rho}^i_{\alpha}{\sigma}^{\alpha}\pd{L}{x^i}+\sum_{r=1}^{k}{\delta}y^{\alpha}_r\pd{L}{y^{\alpha}_r}\\
&={\rho}^i_{\alpha}{\sigma}^{\alpha}\pd{L}{x^i}+{\delta}y^{\alpha}_1\sum_{r=1}^{k}(-1)^{r+1}\frac{d^{r-1}}{dt^{r-1}}\left(\pd{L}{y^{\alpha}_r}\right)+\\
&\qquad+\frac{d}{dt}\left[\sum_{r=2}^k\sum_{j=0}^{r-2}(-1)^j\frac{d^{j}}{dt^{j}}\left(\pd{L}{y^{\alpha}_r}\right){\delta}y^{\alpha}_{r-j-1}\right].
\end{align*}
In the first term, we take into account that ${\delta}y^{\alpha}_r=\dot{{\sigma}}^{\alpha}+C^{\alpha}_{{\beta}{\gamma}}y^{\beta}_1{\sigma}^{\gamma}$ and we perform a further integration by parts. Denoting by ${\pi}_{\alpha}$ the coefficient of ${\delta}y^{\alpha}_1$, 
\begin{equation}\label{momenta}
{\pi}_{\alpha}=\sum_{r=1}^{k}(-1)^{r-1}\frac{d^{r-1}}{dt^{r-1}}\left(\pd{L}{y^{\alpha}_r}\right),
\end{equation}
we obtain 
\begin{align*}
\pai{dL}{{\sigma}^k_a}
&=\left\{{\rho}^i_{\alpha}\pd{L}{x^i}-\Bigl[\frac{d{\pi}_{\alpha}}{dt}+{\pi}_{\gamma}C^{\gamma}_{{\alpha}{\beta}}y^{\beta}_1\Bigr]\right\}{\sigma}^{\alpha}+\\
&\qquad+\frac{d}{dt}\left\{{\pi}_{\alpha}{\sigma}^{\alpha}+\sum_{r=2}^k\sum_{j=0}^{r-2}(-1)^j\frac{d^{j}}{dt^{j}}\left(\pd{L}{y^{\alpha}_r}\right){\delta}y^{\alpha}_{r-j-1}\right\}.
\end{align*}
Relabeling the sums we finally get
\begin{equation}\label{deltaL.plus.thetaL}
\begin{aligned}
\pai{dL}{{\sigma}^k_a}
&=\left\{{\rho}^i_{\alpha}\pd{L}{x^i}-\Bigl[\frac{d{\pi}_{\alpha}}{dt}+{\pi}_{\gamma}C^{\gamma}_{{\alpha}{\beta}}y^{\beta}_1\Bigr]\right\}{\sigma}^{\alpha}+\\
&\qquad+\frac{d}{dt}\left\{{\pi}_{\alpha}{\sigma}^{\alpha}+\sum_{r=1}^k\left[\sum_{j=r}^{k-1}(-1)^{j-r}\frac{d^{j-r}}{dt^{j-r}}\left(\pd{L}{y^{\alpha}_{j+1}}\right)\right]{\delta}y^{\alpha}_r\right\}.
\end{aligned}
\end{equation}
Therefore, the Euler-Lagrange equations take the form of the system of ordinary differential equations 
\begin{equation}\label{EL.equations}
\left\{\begin{aligned}
&\dot{x}^i={\rho}^i_{\alpha}y^{\alpha}_1\\
&\dot{{\pi}}_{\alpha}+{\pi}_{\gamma}C^{\gamma}_{{\alpha}{\beta}}y^{\beta}_1={\rho}^i_{\alpha}\pd{L}{x^i},
\end{aligned}\right.
\end{equation}
where the functions ${\pi}_{\alpha}$ are given by~\eqref{momenta}. 

The coordinate expressions of ${\delta}L$, $\FL$ and ${\theta}_L$ can be obtained from the definition using the  local expressions for the vertical endomorphism. However, from the integration by parts formula we know that the integrand is $\pai{{\delta}L(a^2k)}{{\sigma}}$ and that the boundary terms are $\pai{\FL(a^{2k-1})}{{\sigma}^{(k-1)}}$, and hence from equation~\eqref{deltaL.plus.thetaL} we get that the local expression of ${\delta}L$ is 
\begin{equation}
\label{deltaL.local}
{\delta}L=\Bigl[{\rho}^i_{\alpha}\pd{L}{x^i}-\bigl(d_{\Ti}{\pi}_{\alpha}+{\pi}_{\gamma}C^{\gamma}_{{\alpha}{\beta}}y^{\beta}_1\bigr)\Bigr]e^{\alpha}.
\end{equation}
Taking the induced coordinates $(x^i,y^{\alpha}_1,\ldots,y^{\alpha}_{k-1}\,,\,{\mu}_{\alpha}^0,\ldots,{\mu}_{\alpha}^{k-1})$ on $(\T^EE^k)^*$, the local expression of the Legendre transformation $\map{\FL}{E^{2k-1}}{(\T^EE^{k-1})^*}$ is of the form
\[
\FL(x^i,y^{\alpha}_1,\ldots,y^{\alpha}_{k-1}\,,\,y^{\alpha}_k,\ldots, y^{\alpha}_{2k-1})
=(x^i,y^{\alpha}_1,\ldots,y^{\alpha}_{k-1}\,,\,{\mu}_{\alpha}^0,\ldots,{\mu}_{\alpha}^{k-1}),
\]
where the ${\mu}_{\alpha}^r$ are given by the following functions 
\begin{equation}\label{all.momenta}
{\mu}_{\alpha}^r=\p^r_{\alpha}(x^j,y^{\alpha}_1,\ldots,y^{\alpha}_{2k-1})\equiv\sum_{j=r}^{k-1}(-1)^{j-r}d_{\Ti}^{j-r}\left(\pd{L}{y^{\alpha}_{j+1}}\right).
\end{equation}
Notice that the momenta $\p^r_{\alpha}$ satisfy the relation
\[
\p_{\alpha}^{r-1}=\pd{L}{y^{\alpha}_r}-d_{\Ti}\p^r_{\alpha},\qquad r{\geq}1,
\]
which can serve to define them by recurrence starting from $\p_{\alpha}^{k-1}=\pd{L}{y^{\alpha}_k}$. Also notice that $\p_{\alpha}^0={\pi}_{\alpha}$.

Finally, the coordinate expression of the Cartan form readily follows from that of $\FL$, 
\begin{equation}
{\theta}_L=\sum_{r=0}^{k-1}\p_{\alpha}^r\,\V^{\alpha}_r={\pi}_{\alpha}\,\X^{\alpha}+\sum_{r=1}^{k-1}\p_{\alpha}^r\,\V^{\alpha}_r.
\end{equation}
with $\p_{\alpha}^r$ given by~\eqref{all.momenta}.
\paragraph{Noether's theorem}
An immediate consequence of our variational formalism is the following version of Noether's theorem.
\begin{theorem}
If ${\eta}$ is a section of $E$ such that $d_{{\eta}^k}L=d_{\Ti}F$ for some function $F{\in}\cinfty{E^{k-1}}$, then the function $G{\in}\cinfty{E^{2k-1}}$ given by $G=F-\pai{{\theta}_L}{{\eta}^{2k-1}}$ is a first integral.
\end{theorem}
\begin{proof}
Indeed, we have 
\[
d_{\Ti}G
=d_{\Ti}F-d_{\Ti}\pai{{\theta}_L}{{\eta}^{2k-1}}
=\pai{dL}{{\eta}^k}-\pai{d_{\Ti}{\theta}_L}{{\eta}^{2k}}
=\pai{{\delta}L}{{\eta}}.
\]
Therefore, for any solution $a$ of the Euler-Lagrange equations we have
\[
\frac{d}{dt}(G{\circ}a^{2k-1})=(d_{\Ti}G){\circ}a^{2k}=\pai{{\delta}L(a^{2k})}{{\eta}{\circ}{\gamma}}=0,
\]
so that $G$ is constant along any solution of the Euler-Lagrange equations.
\end{proof}

In particular, if ${\eta}$ is an infinitesimal symmetry of the Lagrangian, i.e.\ $d_{{\eta}^k}L=0$, then the function $\pai{{\theta}_L}{{\eta}^{2k-1}}{\in}\cinfty{E^{2k-1}}$, the momentum in the direction of ${\eta}$, is a first integral.


\section{Examples}
\label{examples}

We will consider in this section some typical examples of Lie algebroids and we will show the form of the Euler-Lagrange equations.


\subsection{The standard case, parameters and quasi-velocities}
In the standard case $E=TM$ taking a coordinate basis $e_i=\pd{}{x^i}$ we recover the standard higher-order Euler-Lagrange equations~\cite{LeRo,CrSaCa,Tu76,PrRo}. The same equations hold when we consider a  system defined by a Lagrangian $L_{\lambda}([{\gamma}]^k)\equiv L({\lambda},[{\gamma}]^k)$, depending on additional parameters ${\lambda}{\in}\Lambda$. In this case the Lie algebroid is $E=\Lambda\times TM\to\Lambda\times M$ with ${\rho}({\lambda},v)=(0_{\lambda},v)$ and the bracket is the bracket of vector fields on $M$ depending on the variables ${\lambda}$ as parameters. 

Also, the formalism developed here allows naturally to use a different local basis $\{e_i={\rho}^j_i\pd{}{x^j}\}$ of vector fields in $M$. In such case the associated coordinates $y^i_1$ are called quasi-velocities and the local expressions for the Euler-Lagrange equations that we have got is the Euler-Lagrange equations written in quasi-velocities, which are sometimes called the \emph{higher-order Hammel equations}. In that expressions, the structure functions $C^i_{jk}$ are the so called Hammel's transpositional symbols, defined by the equation
\[
{\rho}^l_j\pd{{\rho}^i_k}{x^l}-{\rho}^l_k\pd{{\rho}^i_j}{x^l}={\rho}^i_lC^l_{jk}.
\]
For more information about quasi-velocities and their use in Mechanics see~\cite{LMLA, Ca-quasi} and for their use in dynamic optimal control see~\cite{Ab-quasi}. 


\subsection{Systems with holonomic constraints}
Let $E{\subset}TM$ be an integrable subbundle (i.e. an integrable regular distribution on $M$). A curve $\map{a}{{\Real}}{E}$ is admissible if and only if the base curve $\map{{\gamma}={\tau}{\circ}a}{{\Real}}{M}$ is contained into an integral leaf of $E$ and $a=\dot{{\gamma}}$. If we denote by $\calf$ the foliation defined by $E$, so that $E=T\calf$, then it follows that $E^k=T^k\calf$, that is, it is the set of $k$-jets of curves contained in the leaves of $\calf$.

Given a Lagrangian $L{\in}\cinfty{E^k}$ an admissible curve $a=\dot{{\gamma}}$ is a critical point of the action if and only if ${\gamma}$ is a critical point of the restriction of the Lagrangian to the $k$-tangent $T^k(\calf_{{\gamma}(t_0)})$ to the leaf $\calf_{{\gamma}(t_0)}$ which contains ${\gamma}$. This follows from standard optimization results. Therefore the Euler-Lagrange equations can be obtained as the Euler-Lagrange equations for the restriction of the Lagrangian to every leaf, usually called the \emph{holonomic Euler-Lagrange equations}.

In local coordinates $(x^i)=(q^a,w^A)$ adapted to the foliation, so that the leaves are given by the equations $w^A=k^A=\mathrm{constant}$, we have coordinates $(q^a_{(r)},w^A)$, $r=0,\ldots,k$, in $E^k$ and the Euler-Lagrange equations read
\begin{equation}
\label{Euler-Lagrange.holonomic}
\sum_{r=0}^k(-1)^r\frac{d^r}{dt^r}\pd{L}{q^a_{(r)}}=0,\qquad w^A=k^A,
\end{equation}
that is, the standard Euler-Lagrange equations on the variables $q$ depending on $w^A$ as parameters.

Alternatively, one can take a local basis $\{e_a,e_A\}$ of vector fields adapted to the distribution $E$, i.e. $E=\operatorname{span}(\{e^a\})$ and then the Euler-Lagrange equations can be written as $\pai{{\delta}\tilde{L}(a(t)}{e_a}=0$, where $\tilde{L}$ is any extension of $L{\in}\cinfty{E}$ to a function on $TM$. In this way we get an expression of the Euler-Lagrange equations with holonomic constraints written in terms of quasi-velocities.


\subsection{Lagrangian systems on Lie algebras}
A Lie algebra $\g$ can be considered as a Lie algebroid over a singleton $\g\to\{e\}$. The anchor vanishes from where it follows that every curve $\map{{\xi}}{{\Real}}{\g}$ is admissible and hence we have that $\g^k$ is just the cartesian product of $k$ copies of $\g$,
\[
[{\xi}]^{k-1}\equiv\Bigl({\xi}(0),\dot{{\xi}}(0),\ddot{{\xi}}(0),\ldots,\frac{d^{k-1}{\xi}}{dt^{k-1}}(0)\Bigr){\in}\g\times\cdots\times\g.
\] 
A section of $\g\to\{e\}$ is just an element of $\g$ and a local basis $\{e_{\alpha}\}$ of $\g$ provides global coordinates $({\xi}_1,{\xi}_2,\ldots,{\xi}_{k})$ in $\g^k$. Variational vector fields are of the form
\[
{\Xi}_{\xi}{\sigma}=\Bigl({\delta}_{\sigma}{\xi},\frac{d}{dt}{\delta}_{\sigma}{\xi},\ldots,\frac{d^{k-1}}{dt^{k-1}}{\delta}_{\sigma}{\xi}\Bigr)
\qquad\text{with}\quad {\delta}_{\sigma}{\xi}(t)=\dot{{\sigma}}(t)+[{\xi}(t),{\sigma}(t)].
\]
The Euler-Lagrange equations for a Lagrangian $L{\in}\cinfty{\g^k}$ are  
\begin{equation}
\label{Euler-Poincare}
\dot{{\pi}}+\ad^*_{{\xi}}{\pi}=0,
\end{equation}
where ${\pi}$ is given by~\eqref{momenta}, which in the present case takes the global form
\begin{equation}\label{momenta.Lie.algebra}
{\pi}=\sum_{r=1}^{k}(-1)^{r+1}\frac{d^{r-1}}{dt^{r-1}}\left(\frac{{\delta}L}{{\delta}{\xi}_r}\right).
\end{equation}
In this expression $\frac{{\delta}L}{{\delta}{\xi}_r}$ stands for the globally defined partial derivative of $L$ with respect to ${\xi}_r$ given by 
\[
\pai{\frac{{\delta}L}{{\delta}{\xi}_r}({\xi}^0_1,\ldots,{\xi}^0_k)}{v}=\frac{d}{ds}L({\xi}^0_1,\ldots,{\xi}^0_{r-1},{\xi}^0_r+sv,{\xi}^0_{r+1},\ldots,{\xi}^0_k)\at{s=0}.
\]
These equations are called the higher-order Euler-Poincar{\'e} equations~\cite{GBHoetal,CoMa}. They can be interpreted as the equations for parallel transport of the momentum ${\pi}$ with respect to the canonical $\g$-connection on $\g$ defined by $D_{\xi}{\zeta}=[{\xi},{\zeta}]$. See~\cite{PepinFest} for the details in the first order case.

The same kind of equations are obtained in the case of a bundle of Lie algebras $\map{{\tau}}{K}{M}$ where the bracket depends on the variables in $M$, that is, we can have different Lie algebra structures on the fibres $K_m$ for different $m{\in}M$.
 

\subsection{Systems with advected parameters}
We consider a Lie algebra $\g$ acting on a manifold $M$ by means of a morphism of Lie algebras $\g\to\vectorfields{M}$; ${\xi}\mapsto {\xi}_M$. The trivial bundle $\map{{\tau}=\pr_1}{E=M\times\g}{M}$ is endowed with a Lie algebroid structure with anchor ${\rho}(m,{\xi})={\xi}_M(m)$ and where the bracket is induced naturally by the bracket on $\g$ (the bracket of constant sections is just the constant section corresponding to the bracket on $\g$). 

A curve $(m(t),{\xi}(t))$ is admissible if and only if $\dot{m}(t)={\xi}(t)_M(m(t))$. It follows that they are determined by the curve ${\xi}(t)$ and the initial value $m(0)$, and hence we have the identification $E^k\equiv M\times\g^k$ given by
\[
[\,m,{\xi}\,]^{k-1}\equiv \Bigl(m(0),{\xi}(0),\dot{{\xi}}(0),\ldots,\frac{d^{k-1}{\xi}}{dt^{k-1}}(0)\Bigr).
\] 
Variational vector fields are of the form 
\[
{\Xi}_{\xi}{\sigma}=\Bigl({\delta}_{\sigma}m,{\delta}_{\sigma}{\xi},\frac{d}{dt}{\delta}_{\sigma}{\xi},\ldots,\frac{d^{k-1}}{dt^{k-1}}{\delta}_{\sigma}{\xi}\Bigr),
\]
with ${\delta}_{\sigma}m(t)={\rho}(m(t),{\xi}(t))$ and  ${\delta}_{\sigma}{\xi}(t)=\dot{{\sigma}}(t)+[{\xi}(t),{\sigma}(t)]$.

The Euler-Lagrange equations take the form of the so called \emph{Euler-Poincar{\'e} equations with advected parameters}
\begin{equation}
\label{Euler-Poincare.with.advected.parameters}
\left\{\begin{aligned}
&\dot{m}={\xi}_M(m),\\
&\dot{{\pi}}+\ad^*_{{\xi}}{\pi}={\rho}_m^\star\left(\frac{{\delta}L}{{\delta}m}\right),
\end{aligned}\right.
\end{equation}
where ${\pi}$ is given by a expression similar to~\eqref{momenta.Lie.algebra} but depending also on $m$, and ${\rho}_m$ is the restriction $\map{{\rho}_m}{\g}{T_mM}$ of ${\rho}$ to the fibre over $m{\in}M$, i.e.\ ${\rho}_m({\xi})={\xi}_M(m)$. 


\subsection{Systems with symmetry}
We consider a free and proper action of a Lie group $G$ on manifold $Q$, so that $\map{p}{Q}{M=Q/G}$ is a principal bundle. The vector bundle $\map{{\tau}=p{\circ}{\tau}_Q}{E=TQ/G}{M}$ has a natural Lie algebroid structure known as the Atiyah algebroid. The anchor is ${\rho}([v]_G)=Tp(v)$. Sections of $E$ are in one to one correspondence to invariant vector fields on $Q$, and the bracket of two invariant vector fields is also an invariant vector field, defining in this way a bracket on $\sec{E}$. 

A curve ${\xi}(t)$ in $E=TQ/G$ is admissible if and only if there exists a curve $q(t)$ on $q$ such that ${\xi}(t)=[\dot{q}(t)]_G$. It follows that there is a  canonical identification $\map{{\psi}}{T^kQ/G}{(TQ/G)^k}$ given by  ${\psi}([\,[q]^k\,]_G)=[\,[\dot{q}]_G\,]^{k-1}$ (where on the left hand side the action of $G$ on $T^kQ$ is by $k$-jet prolongation of the action of $G$ on $Q$). 

To find some general enough expression for the Euler-Lagrange equations we can consider the case of a trivial bundle $Q=M\times G$, so that $E=TQ/G=TM\times\g$. A curve $(v(t),{\xi}(t))$ in $E$ is admissible if and only if $v(t)=\dot{m}(t)$, where $m$ is the base curve of $v$. Therefore admissible curves are of the form $(\dot{m}(t),{\xi}(t))$ for $m(t)$ a curve in $M$ and ${\xi}(t)$ a curve in $\g$, and we have the identification $E^k\equiv T^kM\times\g^k$, given by 
\[
[\dot{m},{\xi}]^{k-1}\equiv\Bigl([m]^k;{\xi}(0),\dot{{\xi}}(0),\ldots,\frac{d^{k-1}{\xi}}{dt^{k-1}}(0)\Bigr).
\]
Variational vector fields are of the form 
\[
{\Xi}_{(m,{\xi})}({\zeta},{\sigma})=\Bigl({\zeta},\dot{{\zeta}},\ldots,{\zeta}^{(k)};{\delta}_{\sigma}{\xi},\frac{d}{dt}{\delta}_{\sigma}{\xi},\ldots,\frac{d^{k-1}}{dt^{k-1}}{\delta}_{\sigma}{\xi}\Bigr),
\]
with ${\delta}_{\sigma}{\xi}(t)=\dot{{\sigma}}(t)+[{\xi}(t),{\sigma}(t)]$.

A Lagrangian $L$ on $T^kQ/G$, generally defined in terms of a $G$-invariant Lagrangian on $T^kQ$, produces the following set of Euler-Lagrange equations, which are known as the \emph{Lagrange-Poincar{\'e} equations},
\begin{equation}\label{Lagrange-Poincare}
\left\{\begin{aligned}
&\dot{{\pi}}_M=\pd{L}{m}\\
&\dot{{\pi}}+\ad^*_{{\xi}}{\pi}=0,
\end{aligned}\right.
\end{equation}
where ${\pi}$ is given by a expression similar to~\eqref{momenta.Lie.algebra} and ${\pi}_M$ is 
\[
{\pi}_M=\pd{L}{m_{(1)}}-d_{\Ti}\pd{L}{m_{(2)}}+\cdots+(-1)^{k-1}d^{k-1}_{\Ti}\pd{L}{m_{(k)}}.
\]
In other words, the equations look like the ordinary Euler-Lagrange equations on $M$ together with the Euler-Poincar{\'e} equations on $\g$. In the case of a non-trivial bundle one has to use a principal connection on $Q$ to obtain global expressions for the Euler-Lagrange equations (see~\cite{GBHoRa,Co} for the details).


\section{Reduction and reconstruction}
\label{reduction.reconstruction}

It was proved in~\cite{VCLA} that a morphism of Lie algebroids $\map{{\Phi}}{E}{E'}$ defines a map between the path spaces $\map{\hat{{\Phi}}}{\PJ}{\PJ[E']}$ by composition $\hat{{\Phi}}(a)={\Phi}{\circ}a$. If ${\Phi}$ is fiberwise injective (surjective) then $\hat{{\Phi}}$ is an immersion (submersion).  Moreover, variational vector fields are mapped in a simple manner $\T^{\Phi}{\Phi}{\circ}{\Xi}_a{\sigma}={\Xi}_{{\Phi}{\circ}a}({\Phi}{\circ}{\sigma})$. As a consequence, a morphism induces relations between critical points of functions defined on path spaces, in particular between solutions of Euler-Lagrange equations for a first-order Lagrangian $L$ in $E$ and a first-order Lagrangian $L'=L{\circ}{\Phi}$ in $E'$, producing easy reduction and reconstruction results.

This correspondence can be easily extended to the higher-order case as follows. Consider a  Lagrangian $L{\in}\cinfty{E^k}$ and a Lagrangian $L'{\in}\cinfty{E'^k}$ which are related by the map $\map{\Phi^k}{E^k}{E'^k}$, that is, $L=L'\circ\Phi^k$. Then the associated action functionals $S$ on $\PJ$ and $S'$ on $\PJ[E']$ are related by $\hat{\Phi}$, that is,  $S'\circ \hat{\Phi}=S$. Indeed,
\begin{align*}
S'(\hat{\Phi}(a))
&=S'(\Phi\circ a)\\
&=\int_{t_0}^{t_1}(L'\circ(\Phi\circ a)^k)(t)\,dt\\
&=\int_{t_0}^{t_1}(L'\circ\Phi^k\circ a^k)(t)\,dt\\
&=\int_{t_0}^{t_1}(L\circ a^k)(t)\,dt\\
&=S(a).
\end{align*}
For the boundary conditions, if $A'_0={\Phi}^{k-1}(A_0)$ and $A'_1={\Phi}^{k-1}(A_1)$ then $\hat{{\Phi}}$ applies $\PJ_{A_0}^{A_1}$ into $\PJ[E']_{A'_0}^{A'_1}$, so that the corresponding variational problems are related. Indeed, if $a^{k-1}(t_i)=A_i$ then $({\Phi}{\circ}a)^{k-1}(t_i)=\Phi^{k-1}(a^{k-1}(t_i))=A'_i$, for~$i=0,1$.

Reduction theorems easily follows by considering fiberwise surjective morphisms of Lie algebroids. For the sake of clarity, we will omit any reference to the boundary conditions.

\begin{theorem}[Reduction]
\label{reduction}
Let $\map{\Phi}{E}{E'}$ be a fiberwise surjective morphism of Lie algebroids. Consider a Lagrangian $L$ on $E^k$ and a Lagrangian $L'$ on $E'^k$ such that $L=L'\circ\Phi^k$. If $a$ is a solution of the Euler-Lagrange equations for $L$ then $a'=\Phi\circ a$ is a solution of Euler-Lagrange equations for $L'$.
\end{theorem}
\begin{proof}
Since $S'\circ \hat{\Phi}=S$ we have that $\pai{\d S'(\hat{\Phi}(a))}{T_a\hat{\Phi}(v)} =\pai{\d S(a)}{v}$ for every $v\in T_a\PJ_{A_0}^{A_1}$. If $\Phi$ is fiberwise surjective, then $\hat{\Phi}$ is a submersion, from where it follows that $\hat{\Phi}$ maps critical points of $S$ into critical points of $S'$, i.e. solutions of the Euler-Lagrange  equations for $L$ into solutions of Euler-Lagrange equations for~$L'$.
\end{proof}

We can reduce partially a system and then reduce again the obtained system. The final result obviously coincides with the system obtained by the total reduction.

\begin{theorem}[Reduction by stages]
\label{reduction.by.stages}
Let $\map{\Phi_1}{E}{E'}$ and $\map{\Phi_2}{E'}{E''}$ be fiberwise surjective morphisms of Lie algebroids. Let $L$, $L'$ and $L''$ be Lagrangian functions on $E^k$, $E'^k$ and $E''^k$, respectively, such that $L'\circ\Phi_1^k=L$ and $L''\circ\Phi_2^k=L'$. Then the result of reducing first by $\Phi_1$ and later by $\Phi_2$ coincides with the reduction by $\Phi=\Phi_2\circ\Phi_1$.
\end{theorem}
\begin{proof}
It is obvious since $\Phi=\Phi_2\circ\Phi_1$ is also a fiberwise surjective morphism of Lie algebroids.
\end{proof}

As an example of the above situation we can consider a system with a group of symmetry $G$. If $N$ is closed normal subgroup of $G$ we can reduce first by $N$ and later by $G/N$. Alternatively we can reduce directly by the full group $G$. The result of both procedures is the same. See~\cite{CeMaRa,NHLSLA} for the first order case.

\bigskip

For the reconstruction of solutions we have the following result, which is immediate from the variational formalism.
\begin{theorem}[Reconstruction]
\label{reconstruction}
Let $\map{\Phi}{E}{E'}$ be a morphism of Lie algebroids. Consider a Lagrangian $L$ on $E^k$ and a Lagrangian $L'$ on $E'^k$ such that $L=L'\circ\Phi^k$. If $a$ is an $E$-path and $a'=\Phi\circ a$ is a solution of the Euler-Lagrange equations for $L'$ then $a$ itself is a solution of the Euler-Lagrange equations for $L$.
\end{theorem}
\begin{proof}
Since $S'\circ \hat{\Phi}=S$ we have that $\pai{\d S'(\hat{\Phi}(a))}{T_a\hat{\Phi}(v)} =\pai{\d S(a)}{v}$ for every $v\in T_a\PJ_{A_0}^{A_1}$. If $\hat{\Phi}(a)$ is a solution of Lagrange's equations for $L'$ then $\d S'(\hat{\Phi}(a))=0$, from where it follows that $\d S(a)=0$.
\end{proof}

The reconstruction procedure can be understood as follows. Consider a fiberwise surjective morphism $\map{\Phi}{E}{E'}$ and the associated reduction map $\map{\hat{\Phi}}{\PJ}{\PJ[E']}$. Given an $E'$-path $a'\in\PJ[E']$ solution of the dynamics defined by the Lagrangian $L'$, we look for an $E$-path $a\in\PJ$ solution of the dynamics for the Lagrangian $L=L'\circ\Phi$, such that $a'=\hat{\Phi}(a)$. For that, it is sufficient to find a map $\map{{\Upsilon}}{\PJ[E']}{\PJ}$ such that $\hat{\Phi}\circ{\Upsilon}=\id_{\PJ[E']}$. Indeed, given the $E'$-path $a'$ solution for the reduced Lagrangian $L'$, the curve $a={\Upsilon}(a')$ is an $E$-path and satisfy $\Phi\circ a=a'$. From Reconstruction Theorem~\ref{reconstruction} we deduce that $a$ is a solution of the Euler-lagrange equations for the original Lagrangian. 

When specifying the boundary conditions is necessary, the map ${\Upsilon}$ must restrict to a map $\map{{\Upsilon}_{A_0}^{A_1}}{\PJ[E']_{A'_0}^{A'_1}}{\PJ[E]_{A_0}^{A_1}}$ where $A'_0={\Phi}^{k-1}(A_0)$ and $A'_1={\Phi}^{k-1}(A_1)$.

Of course, one can define several maps ${\Upsilon}$, and different maps will produce different $E$-paths $a$ for the same $E'$-path $a'$. In most cases of interest, ${\Phi}$ is fiberwise bijective, so that $\hat{{\Phi}}$ is a local diffeomorphism, and the reconstruction process can be done with the help of some global diffeomorphism $\map{\bar{{\Phi}}}{\PJ}{P\times\PJ[E']}$, for some manifold $P$, such that $\pr_2{\circ}\bar{{\Phi}}=\hat{{\Phi}}$. Then fixing $p{\in}P$ we define ${\Upsilon}_p(a')=\bar{{\Phi}}^{-1}(p,a')$. In this way the set of solutions to which $a'$ can be lifted is parametrized by $P$. See bellow for some explicit examples of application of this procedure (and some extensions).

\bigskip

Next, we present some examples where the reduction process indicated above can be applied.


\subsection{Euler-Poincar{\'e} reduction}
Consider a Lie group $G$ and its Lie algebra~$\g$. The map $\map{\Phi}{TG}{\g}$ given by $\Phi(g,\dot{g})=g^{-1}\dot{g}$ is a fiberwise bijective morphism of Lie algebroids. The $k$-jet prolongation $\map{{\Phi}^k}{T^kG}{\g^k}$ of ${\Phi}$ is ${\Phi}^k([g]^k)=T^k\ell_{g^{-1}}[g]^k=[g^{-1}\dot{g}]^{k-1}$, which can be explicitly written as 
\[
{\Phi}^k(g,\dot{g},\ldots,g^{(k)})=({\xi},\dot{{\xi}},\ldots,{\xi}^{(k-1)})\qquad\text{with ${\xi}=g^{-1}\dot{g}$.}
\] 
Assume that $L$ is a left-invariant Lagrangian function on $T^kG$, i.e.
\[
L(T^k\ell_{h^{-1}}[g]^k)=L([h^{-1}g]^k)=L([g]^k)
\]
for every element $h{\in}G$ and every curve $g(t)$ on $G$. Then we can define a Lagrangian $L'$ on $\g^k$ by means of 
\[
L'([{\xi}]^{k-1})=L([g_e]^k)
\]
where $g_e$ is the solution of the differential equation $\dot{g}=T\ell_g{\xi}(t)$ with initial value $g_e(0)=e$. In other words we have $L([g]^k)=L'([g^{-1}\dot{g}]^{k-1})$ for every curve $g(t)$ in $G$, or equivalently $L=L'{\circ}{\Phi}^k$. It follows that if $a$ is a solution of the higher-order Euler-Lagrange for $L$ in the Lie group $G$ then $a'={\Phi}{\circ}a$ is a solution of the higher-order Euler-Poincare for $L'$ on $\g$. Moreover, since $\Phi$ is fiberwise bijective every solution can be found in this way, so that the higher-order Euler-Lagrange equations on the group reduce to the higher-order Euler-Poincar{\'e} equations on the Lie algebra.

For the reconstruction process we note that the map $\map{\bar{{\Phi}}}{\PJ[TG]}{G\times\PJ[\g]}$ given by $\bar{{\Phi}}(g,\dot{g})=(g(t_0), T\ell_{g^{-1}}\dot{g})$ is a global diffeomorphism. Its inverse is the map $\bar{{\Phi}}^{-1}(g_0,{\xi})=(g,\dot{g})$ where $g(t)$ is the integral curve of the left-invariant time-dependent vector field $X(t,g)=T\ell_g{\xi}(t)$, with initial value $g(t_0)=g_0$. Thus the map ${\Upsilon}_{g_0}({\xi})=\bar{{\Phi}}^{-1}(g_0,{\xi})$ provides a reconstruction map. Notice that the curve $g(t)$ above is $g(t)=g_0g_e(t)$ where $g_e(t)$ is the integral curve of $X$ with $g_e(t_0)=e$.


\subsection{Lie groupoid reduction}
For integrable Lie algebroids, the variational principle developed here can be obtained via a reduction of an ordinary higher-order variational principle with holonomic constraints. In the first-order case this was the argument used by Weinstein~\cite{Weinstein} to obtain a variational principle on (integrable) Lie algebroids. Here we consider the higher-order case.

Consider a Lie groupoid $\G$ over $M$ with source $\s$ and target $\t$,  and denote by $E$ its Lie algebroid, $E=\cala(\G)$. Denote by $T^{\s}\G\to\G$ the kernel of $T\s$ with the structure of Lie algebroid as integrable subbundle of $T\G$. The integral manifolds of $T^{\s}\G$ are the $\s$-fibres of the groupoid, so that $T^{\s}\G$ is the distribution tangent to the foliation $\cals={\cup}_{m{\in}M}\s^{-1}(m)$, and hence $(T^{\s}\G)^k=(T\cals)^k=T^k\cals$. 

The map $\map{\Phi}{T^{\s}\G}{E}$ given by left translation to the identities, $\Phi(v_g)=T\ell_{g^{-1}}(v_g)$ is a fiberwise bijective morphism of Lie algebroids. As a consequence, if $L$ is a Lagrangian function on $E^k$ and  $\boldsymbol{L}$ is the associated left invariant Lagrangian on $(T^{\s}\G)^k$, then the solutions of the higher-order Euler-Lagrange equations for $\boldsymbol{L}$ (which are the holonomic Euler-Lagrange equations) project by $\Phi$ to solutions of the higher-order Euler-Lagrange equations on the Lie algebroid $E$.

For the reconstruction process we consider the manifold $\G\times_M\PJ$ defined by  $\G\times_M\PJ=\set{(g_0,a){\in}\G\times\PJ}{\t(t_0)={\tau}(a(t_0))}$. The map $\map{\bar{{\Phi}}}{\PJ[T^{\s}\G]}{\G\times_M\PJ}$ given by $\bar{{\Phi}}(g,\dot{g})=(g(t_0), T\ell_{g^{-1}}\dot{g})$ is a global diffeomorphism. Its inverse is given by $\bar{{\Phi}}^{-1}(g_0,a)=(g,\dot{g})$ where $g(t)$ is the integral curve of the left-invariant time-dependent vector field $X(t,g)=T\ell_{g}a(t)$ with initial value $g(t_0)=\boldsymbol{{\epsilon}}(\t(g_0))$ (the identity at the point $\t(g_0)$). Thus the map ${\Upsilon}_{g_0}({\xi})=\bar{{\Phi}}^{-1}(g_0,{\xi})$ provides a reconstruction map. 

Notice that in this occasion ${\Upsilon}_{g_0}$ is defined as a map  $\map{{\Upsilon}_{g_0}}{\PJ_{m_0}}{\PJ[T^{\s}\G]_{g_0}}$ where $\PJ[T^{\s}\G]_{g_0}=\set{(g,\dot{g}){\in}\PJ}{g(t_0))=g_0}$, $m_0=\t(g_0)$ and $\PJ_{m_0}=\set{a{\in}\PJ}{{\tau}(a(t_0))=m_0}$.


\subsection{Euler-Poincar{\'e} reduction with advected parameters}
Let $G$ be a Lie group acting from the right on a manifold $M$. We consider the Lie algebroid $E=M\times TG\to M\times G$ where $M$ is a parameter manifold, that is, the anchor is $\rho(m,v_g)=(0_m,v_g)$ and the bracket is determined by the standard bracket of vector fields on $G$, i.e. of sections of $TG\to G$, depending on the coordinates in $M$ as parameters. Consider also the transformation Lie algebroid $E'=M \times \g\to M$, where $\rho(m,\xi)=\xi_M(m)$, ($\xi_M$ being the fundamental vector field associated to $\xi\in\g$). The map $\map{{\Phi}}{M\times TG}{M\times\g}$ given by $\Phi(m,v_g)=(mg,g^{-1}v_g)$ is a morphism of Lie algebroids over the action map $\varphi(m,g)=mg$, and it is fiberwise bijective.

Consider a Lagrangian  $L$ on $M\times T^kG$ depending on the elements of $M$ as parameters $L(m,[g]^k)=L_m([g]^k)$. Assume that $L$ is not left invariant but obeys to the following transformation rule 
\[
L(m,T^k\ell_h[g]^k)=L(m,[hg]^k)=L(mh,[g]^k),
\]
for every element $h{\in}G$ and every curve $g(t)$ in $G$ (equivalently $L$ is invariant by the joint left action $L(mh^{-1},[hg]^k)=L(m,[g]^k)$).

We consider the Lagrangian $L'$ on $E'^k$ given by $L'(m,[{\xi}]^{k-1})=L(m,[g_e]^k)$, where $g_e(t)$ is the solution of the initial value problem $\dot{g}=T\ell_g{\xi}(t)$, $g(t_0)=e$. Then $L'\circ\Phi^k=L$. The variables $m$ which initially are parameters are now dynamic variables due to the group action and are called advected parameters. In this way, solutions of the higher-order Euler-Lagrange equations for $L$ (standard Euler-Lagrange equations with parameters) are mapped by $\Phi$ to solutions of the higher-order Euler-Lagrange equations for $L'$, i.e.\ the higher-order Euler-Poincar{\'e} equations with advected parameters.

For the reconstruction process we consider the global diffeomorphism $\map{\bar{{\Phi}}}{\PJ[M\times TG]}{G\times\PJ[M\times\g]}$ given by $\bar{{\Phi}}\bigl(m, (g,\dot{g})\bigr)=\bigl(g(t_0), (mg,T\ell_{g^{-1}}\dot{g})\bigr)$. Its inverse is $\bar{{\Phi}}^{-1}\bigl(g_0,(m,{\xi})\bigr)=\bigl(mg^{-1},(g,\dot{g})\bigr)$ where $g(t)$ is the integral curve of the left-invariant time-dependent vector field $X(t,g)=T\ell_{g}{\xi}(t)$ with $g(t_0)=g_0$. Thus the map ${\Upsilon}_{g_0}(m,{\xi})=\bar{{\Phi}}^{-1}\bigl(g_0,(m,{\xi})\bigr)$ provides a reconstruction map.


\subsection{Lagrange-Poincar{\'e} reduction}
We consider a Lie group $G$ acting free and properly on a manifold $Q$, so that the quotient map $\map{p}{Q}{M=Q/G}$ has the structure of a principal bundle. We consider the standard Lie algebroid structure on $E=TQ$ and the Atiyah algebroid $E'=TQ/G\to M$. The quotient map $\map{\Phi}{TQ}{TQ/G}$, $\Phi(v)=[v]_G$ is a Lie algebroid morphism and it is fiberwise bijective. Every $G$-invariant Lagrangian on $T^kQ$ defines uniquely a Lagrangian $L'$ on $E'^k$ such that $L'\circ\Phi^k=L$, explicitly given by
\[
L'([\,[\dot{q}]_G\,]^{k-1})=L'([\,[q]^k\,]_G)=L([q^k]).
\] 
Therefore every solution of the $G$-invariant Lagrangian on $T^kQ$ projects to a solution of the reduced Lagrangian on $(TQ/G)^k\equiv T^kQ/G$, and every solution on the reduced space can be obtained in this way. Thus, the higher-order Euler-Lagrange equations on the principal bundle reduce to the higher-order Lagrange-Poincar{\'e} equations on the Atiyah algebroid~\cite{GBHoRa}.

For the reconstruction process we consider the manifold $Q\times_M\PJ[TQ/G]$ defined by $Q\times_M\PJ[TQ/G]=\set{(q_0,a){\in}Q\times\PJ[TQ/G]}{p(t_0)={\tau}(a(t_0))}$. The map $\map{\bar{{\Phi}}}{\PJ[TQ]}{Q\times_M\PJ[TQ/G]}$ given by $\bar{{\Phi}}(q,\dot{q})=(q(t_0), [\dot{q}]_G)$ is a global diffeomorphism. Its inverse is given by $\bar{{\Phi}}^{-1}(q_0,a)=(q,\dot{q})$ where $q(t)$ is the curve determined as follows. The curve $a(t)$ is admissible, so that it is of the form $a(t)=[\bar{q}(t),\dot{\bar{q}}(t)]_G$; if $g_0{\in}G$ is the unique element in the group such that $\bar{q}(t_0)=g_0q_0$, then we define $q(t)=g_0^{-1}\bar{q}(t)$. 
Therefore the map ${\Upsilon}_{q_0}(a)=\bar{{\Phi}}^{-1}(q_0,a)$ provides a reconstruction map, which in this case is  defined as a map  $\map{{\Upsilon}_{q_0}}{\PJ[TQ/G]_{m_0}}{\PJ[TQ]_{q_0}}$ where $m_0=p(q_0)$. 

\section{Conclusions and outlook}

In this paper we have obtained the Euler-Lagrange equations for a higher-order Lagrangian. The emphasis has been on the variational description using the tools of variational calculus on Lie algebroids developed in~\cite{VCLA}.

However, there are other many aspects of the theory that we have left out and will be studied in future. The variational structure strongly suggests that a symplectic formalism is possible, similar to Klein's formalism~\cite{Klein} of standard Lagrangian Mechanics and its generalization to Lie algebroids~\cite{LMLA}, and the corresponding Hamiltonian version~\cite{Medina, LSDLA}. It is interesting to study such formalisms and the transformation properties of the symplectic form.

On the other hand, a geometric version of Pontryagin maximum principle which allows reduction in the setting of Lie algebroids was studied in~\cite{ROCT, SIGMA} (a direct proof was given in~\cite{GrJo}, based on needle variations~\cite{Pontryagin} and the results in~\cite{VCLA}). It is interesting to  study the relation between that results with those on this paper and the proposed Hamiltonian versions already mentioned, on one hand, and with other `variational principles' that appeared recently in the literature (Hamilton-Pontryagin~\cite{YoMa}, Clebsh-Pontryagin~\cite{GBRa}, etc.), on the other.
 
Concerning the optimality properties of the solutions, in this paper we have obtained only first order conditions, that is conditions for critical points of the action. To characterize those which are local maxima/minima we will study further optimality conditions. Since our results are based on a \textsl{bona fide} variational principle it is expected that the Hessian of the action contains such information. Moreover, when two Lagrangians are related by a morphism the Hessians at corresponding critical points should be also related.

Problems with time dependent Lagrangians can be treated by embedding $E$ into $T\Real\times E$ as it is mentioned in~\cite{VCLA}. However the formalism introduced in~\cite{LASLSAB, GrGrUr-Aff} for first-order systems can be easily generalized to the present case.

All this problems and the possible generalizations to field theory, following the ideas in~\cite{CFTLAVA}, will be studied elsewhere.


\end{document}